\documentclass[10pt,reqno]{amsart}
\usepackage{bbm}
\usepackage{tabu}
\usepackage{amsmath}
\usepackage{mathrsfs}
\usepackage{bm}
\usepackage{amsfonts} 
\usepackage{graphicx,latexsym,bm,amsmath,amssymb,verbatim,multicol,lscape}
\usepackage{enumerate}
\newtheorem{thm}{Theorem} [section]

\newtheorem{cor}[thm]{Corollary}
\newtheorem{lem}[thm]{Lemma}
\newtheorem{prop}[thm]{Proposition}

\theoremstyle{definition}
\newtheorem{defn}[thm]{Definition}
\theoremstyle{remark}
\newtheorem{rem}[thm]{Remark}

\numberwithin{equation}{section}

\newcommand{\fq}{{\mathbb F}_{q}}
\newcommand{\fp}{{\mathbb F}_{p}}

\newcommand{\fth}{{\mathbb F}_{3}}

\newcommand{\Tr}{\mbox{Tr}}

\newcommand{\rmv}[1]{}

\def\<{\left\langle}
\def\>{\right\rangle}

\begin{document}

\title[A class of ternary codes with few weights]
{A class of ternary codes with few weights}%
\author{Kaimin Cheng}
\address{School of Mathematics and Information, China West Normal University,
Nanchong, 637002, P. R. China}
\email{ckm20@126.com,kcheng1020@gmail.com}
\thanks{This work was supported partially by the Natural Science Foundation of China (No. 12226335) and the China Scholarship Council Foundation(No. 202301010002).}
\keywords{Linear codes, Weight distribution, Optimal codes, Exponential sums, Weil bound.}
\subjclass[2000]{Primary 94B05}
\date{\today}
\begin{abstract}
Let \(\ell^m\) be a power with \(\ell\) a prime greater than 3 and \(m\) a positive integer such that 3 is a primitive root modulo \(2\ell^m\). Let \(\mathbb{F}_3\) be the finite field of order 3, and let $\mathbb{F}$ be the \(\ell^{m-1}(\ell-1)\)-th extension field of \(\mathbb{F}_3\). Denote by \(\text{Tr}\)  the absolute trace map from $\mathbb{F}$ to \(\mathbb{F}_3\). For any \(\alpha \in \mathbb{F}_3\) and $\beta \in\mathbb{F}$, let \(D\) be the set of nonzero solutions in $\mathbb{F}$ to the equation \(\text{Tr}(x^{\frac{q-1}{2\ell^m}} + \beta x) = \alpha\). In this paper, we investigate a ternary code \(\mathcal{C}\) of length \(n\), defined by \(\mathcal{C} := \{(\text{Tr}(d_1x), \text{Tr}(d_2x), \dots, \text{Tr}(d_nx)) : x \in \mathbb{F}\}\) when we rewrite \(D = \{d_1, d_2, \dots, d_n\}\). Using recent results on explicit evaluations of exponential sums, the Weil bound, and combinatorial techniques, we determine the Hamming weight distribution of the code \(\mathcal{C}\). Furthermore, we show that when \(\alpha = \beta = 0\), the dual code of \(\mathcal{C}\) is optimal with respect to the Hamming bound.
\end{abstract}

\maketitle

\section{Introduction}
Let $p$ be a prime number and $q$ a power of $p$. Denote by $\fq$ the finite field with $q$ elements and by $\fq^*=\fq\setminus\{0\}$ the set of nonzero elements in $\fq$. We denote by ${\rm Tr}$ the absolute \textit{trace} function from $\mathbb{F}_q$ onto $\mathbb{F}_p$. Let $n,k$ and $d$ be positive integers. An $[n,k,d]$ $p$-ary linear code is a $k$-dimensional subspace of $\mathbb{F}_p$ with minimum (Hamming) distance $d$. For a positive integer $i$, denote by $A_i$ the number of codewords with Hamming weight $i$ in a code $\mathcal{C}$ of length $n$.
Let
$$1+A_1z+A_2z^2+\cdots+A_nz^n$$
be the \textit{weight enumerate} of $\mathcal{C}$.
We call the sequence $(A_1,A_2,\ldots,A_n)$ the\textit{ weight distribution} of
$\mathcal{C}$. If the number of nonzero $A_i$ in the sequence $(A_1,A_2,\ldots,A_n)$ is $N$, then the code $\mathcal{C}$ is called an $N$-weight code. The weight distribution determines both the minimum distance of the code and its error-correcting capability. It also contains key information on computation of the probability of error detection and correction with respect to some error detection and correction algorithms \cite{[Klo]}. As a result, the study of the weight distribution of a linear code an important research topic in coding theory. Given a set $D=\{d_1,d_2,\ldots,d_n\}\subseteq\mathbb{F}_q^*$, we define $\mathcal{C}_D$ as
\begin{align}\label{c1-1}
\mathcal{C}_D:=\{({\rm Tr}(d_1x),{\rm Tr}(d_2x),\ldots,{\rm Tr}(d_nx)):\ x\in\mathbb{F}_q\}.
\end{align}
Clearly, $\mathcal{C}_D$ is a linear codes of length $n$ over $\mathbb{F}_p$, and we call $D$ the \textit{defining set} of the code $\mathcal{C}_D$. This defining set construction is crucial, as any linear code over $\fp$  can be represented as for some defining set $D$ (which may be a multiset) \cite{[Din]}. Several classes of codes with few weights have been constructed by appropriately choosing the defining set $D\subseteq\mathbb{F}_q$ (see, for instance, \cite{[Din1],[Din2],[TXF],[WDX]}).

The computation of explicit values of exponential sums over finite fields plays a crucial role in the construction of linear codes. The problem of explicitly evaluating the Weil sums with specific forms has been investigated extensively (see, for example, \cite{[Car1],[Car2],[Cou],[FH],[Moi],[Moi2],[Wan1],[Wan2]}), but it is still quite difficult in general. Recently, the author and Gao \cite{[CG]} evaluated a class of binomial Weil sums, from which they obtained a class of two-weight codes and determined the weight enumerates. In the paper, based on the results of binomial Weil sums \cite{[CG]}, by using the Weil bound together with some combinatorial technique, a class of ternary codes, that is, $p=3$, with few weights is constructed, and the weight distribution is settled. Additionally, we determine the parameters of the dual codes, identifying some of them as optimal. This work significantly extends the results of \cite{[CG]} and \cite{[WDX]} to a broader context, and the linear codes with few weights presented in this paper have applications in secret sharing \cite{[ADHK]}, authentication codes \cite{[DW]}, combinatorial designs \cite{[CK]}, etc.

\section{Preliminaries}
This section provides the necessary preliminaries and background results.. Let $p$ be an odd prime.  Let $m$ be a positive integer and $\ell$ be an odd prime not divisible by $p$ such that $p$ is a primitive root modulo $2\ell^m$. Let $\fq$ be the finite field of order $q$, where $q=p^{e}$ with $e=\phi(2\ell^m)=(\ell-1)\ell^{m-1}$. Let $g$ be a fixed primitive element of $\fq$ and $\xi=g^{\frac{q-1}{2\ell^m}}$. Let $\chi$ be the canonical additive character of $\fq$, i.e., $\chi(c)=e^{\frac{2\pi i{\rm Tr}(c)}{p}}$ for any $c\in\fq$. For any $a,b\in\fq$, define two exponential sums over $\fq$ by
\begin{align}\label{c2-1-1}
S_{2\ell^m}(a,b):=\sum_{x\in\fq^*}\chi(ax^{\frac{q-1}{2\ell^m}}+bx)
\end{align}
and
$$S_{2\ell^m}(a):=\sum_{i=0}^{2\ell^m-1}\chi(a\xi^{i}).$$
One readily finds that
\begin{align}\label{c2-1}
S_{2\ell^m}(a,0)=\frac{q-1}{2\ell^m}S_{2\ell^m}(a).
\end{align}
So, the evaluation of $S_{2\ell^m}(a,b)$ can be derived from $S_{2\ell^m}(a)$ immediately if $b=0$. For the case of $b\ne 0$, our previous results \cite{[CG]} provide a way to get the values of $S_{2\ell^m}(a,b)$, which leads us to some more precise results (see Section 3).
\begin{lem}\cite{[CG]}\label{lem2.1}
For any $a\in\fq$ and $b\in\fq^*$,
$$S_{2\ell^m}(a,b)=\chi(ab^{-\frac{q-1}{N}})\sqrt{q}-
\frac{\sqrt{q}+1}{{2\ell^m}}
S_{2\ell^m}(ab^{-\frac{q-1}{N}}).$$
\end{lem}
Note that $\xi=g^{\frac{q-1}{2\ell^m}}$, and one checks that $\fq=\fp(\xi)$. It follows that $\{\xi,\xi^2,\ldots,\xi^{e}\}$ is a basis of $\fq$ over $\fp$. Then all elements of $\fq$ can be $\fp$-linearly represented by $\{\xi,\xi^2,\ldots,\xi^{e}\}$.
\begin{lem}\cite{[CG]}\label{lem2.2}
For any $a\in\fq$, if we write $a=\sum_{i=1}^ea_i\xi^i$ with each $a_i\in\fp$, then
$$S_{2\ell^m}(a)=\sum_{i=0}^{\ell^{m-1}-1}\Big(\zeta_p^{\Delta
_i}+\zeta_p^{-\Delta_i}+\sum_{t=1}^{\ell-1}\left(\delta_{i,t}+
\delta_{i,t}^{-1}\right)\Big),$$
where
\begin{align*}
\Delta_i=-\ell^{m-1}\sum_{k=1}^{\ell-1}(-1)^ka_{k\ell^{m-1}-i},\  \delta_{i,t}=\zeta_p^{(-1)^t\Delta_i+\ell^ma_{t\ell^{m-1}-i}},
\end{align*}
and $\zeta_p$ represents a primitive $p$-th root of unity.
\end{lem}
So for the given $a\in\fq$, the significant step in process of evaluating $S_{2\ell^m}(a)$ is to figure out all representation coefficients $a_i$ of $a$. For convenience, if we write $a=\sum_{i=1}^ea_i\xi^i$ with each $a_i\in\fp$, then one lets
$$a^{(j)}:=(a_{\ell^{m-1}-j},a_{2\ell^{m-1}-j},\ldots,
a_{(\ell-1)\ell^{m-1}-j})\in\mathbb{F}_p^{\ell-1}
$$
for any $0\le j\le \ell^{m-1}-1$.
\begin{cor}\label{cor2.3}
For any $a\in\fp$, we have
$$S_{2\ell^m}(a)=\zeta_p^{\ell^{m-1}(\ell-1)a}+
\zeta_p^{-\ell^{m-1}(\ell-1)a}+(\ell-1)(\zeta_p^{\ell^{m-1}a}+
\zeta_p^{-\ell^{m-1}a})+2\ell^m-2\ell,$$
where $\zeta_p$ is a primitive $p$-th root of unity.
In particular, if $p=3$, then
\begin{align}\label{c2-2}
S_{2\ell^m}(1)=S_{2\ell^m}(2)=\begin{cases}
2\ell^m-3\ell+3,&\text{if}\ \ell\equiv 1\pmod{3},\\
2\ell^m-3\ell,&\text{if}\ \ell\equiv -1\pmod{3},\\
\end{cases}
\end{align}
which are positive integers.
\end{cor}
\begin{proof}
First, one sees that $\xi$ is a primitive $2\ell^m$-th root of unity. It implies that $Q_{2\ell^m}(\xi)=0$, where $Q_{2\ell^m}(x)$ is the $2\ell^m$-th cyclotomic polynomial over $\fp$. Note that
$$Q_{\ell^m}(x)=\sum_{k=0}^{\ell-1}x^{k\ell^{m-1}}\ \text{and}\
Q_{2\ell^m}(x)=Q_{\ell^m}(-x).$$
It follows that
\begin{align}\label{c2-3}
1=\sum_{k=1}^{\ell-1}(-1)^{k-1}\xi^{k\ell^{m-1}}.
\end{align}
Then $a=\sum_{k=1}^{\ell-1}(-1)^{k-1}a\xi^{k\ell^{m-1}}$,
that is,
$$a^{(0)}=(\underbrace{a,-a,\ldots,a,-a}_{\ell-1}),\ \text{and}\
a^{(j)}=(\underbrace{0,0,\ldots,0,0}_{\ell-1})$$
is an all zero sub-vector for each $1\le j\le \ell^{m-1}-1$. Then one computes that
$$
\Delta_j=\begin{cases}
\ell^{m-1}(\ell-1)a,&\text{if}\ j=0,\\
0,&\text{if}\ j\ne 0,
\end{cases}
\ \text{and}\ \delta_{j,t}=\begin{cases}
\zeta_p^{(-1)^{t-1}\ell^{m-1}a},&\text{if}\ j=0,\\
1,&\text{if}\ j\ne 0,
\end{cases}
$$
where $\Delta_j$ and $\delta_{j,t}$ are defined as in Lemma \ref{lem2.2}. Therefore, one obtains from Lemma \ref{lem2.2} that
$$S_{2\ell^m}(a)=\zeta_p^{\ell^{m-1}(\ell-1)a}+
\zeta_p^{-\ell^{m-1}(\ell-1)a}+(\ell-1)(\zeta_p^{\ell^{m-1}a}+
\zeta_p^{-\ell^{m-1}a})+2\ell^m-2\ell,$$
as desired.
\end{proof}
The following result from \cite{[CG]} on the trace of any power $\xi^j$  will be used in Section 3.
\begin{lem}\label{lem2.4}
For any nonnegative integer $j$, we have
$${\rm Tr}(\xi^j)=\begin{cases}
(\ell-1)\ell^{m-1},& \text{if}\ 2\ell^m\mid j,\\
-(\ell-1)\ell^{m-1},& \text{if}\ \ell^m\mid j\ \text{but}\ 2\nmid j,\\
-\ell^{m-1},& \text{if}\ 2\ell^{m-1}\mid j\ \text{but}\ \ell^m\nmid j,\\
\ell^{m-1},& \text{if}\ \ell^{m-1}\mid j\ \text{but}\ \ell^m\nmid j,\ 2\nmid j,\\
0,& \text{otherwise}.
\end{cases}$$
\end{lem}
At the end of this section, we present an important result of the existence of elements satisfying certain conditions.
\begin{lem}\label{lem2.5}
Let $q\ne 3$ be an odd prime power, $\fq$ be the finite field with $q$ elements and $\fq^*=\fq\setminus\{0\}$. Let $g$ be a primitive element of $\fq$. For any $\alpha\in\fq^*$, denote ${\rm Ind}(\alpha)$ to be the least nonnegative integer $k$ such that $\alpha=g^k$. Let $d$ be a positive divisor of $q-1$. Let $C_j$ be the $j$-th cyclotomic class of order $d$ defined by $C_j=\{g^jx^d:\ x\in\fq^*\}$ with $0\le j\le d-1$. Let $\beta$ be a given nonzero element of $\fq$. We have that if $d\le d_0:=(2-\frac{4}{\sqrt{5}})^{1/2}q^{1/4}\approx 0.4595q^{1/4}$ then
for any two integers $0\le j_1,j_2\le q-2$ there exists an element $\gamma\in\fq^*\setminus\{\pm\beta\}$ such that $\beta+\gamma\in C_{j_1}$ and $\beta-\gamma\in C_{j_2}$.
\end{lem}
\begin{proof}
It is trivial for $d=1$. In the following we always assume that $d\ge 2$. Let $j_1$ and $j_2$ be the two integers with $0\le j_1,j_2\le q-2$. Denote by $N_{j_1,j_2}$ the number of elements $\gamma\in\fq^*\setminus\{\pm\beta\}$ satisfying $\beta+\gamma\in C_{j_1}$ and $\beta-\gamma\in C_{j_2}$, that is,
$$N_{j_1,j_2}=\#\{\gamma\in\fq^*\setminus\{\pm\beta\}:\ \beta+\gamma\in C_{j_1}\ \text{and}\ \beta-\gamma\in C_{j_2}\}.$$
So, it is sufficient to prove that $N_{j_1,j_2}>0$ when $d\le d_0$ which will be done in what follows. Let $\psi$ be a multiplicative character of $\fq$ of order $d$. By definition of cyclotomic classes, one then has that $$
C_j=\{y\in\fq^*:\ {\rm Ind}(y)\equiv j\pmod{d}\}=\{y\in\fq^*:\ \psi(yg^{-j})=1\}$$
for any integer $j$ with $0\le j\le d-1$. Note that for each $j$
$$\sum_{i=0}^{d-1}\psi^i(xg^{-j})=\begin{cases}
d,&x\in C_j\\
0,&x\in\fq^*\setminus C_j.
\end{cases}$$
It then follows that
\begin{align}\label{c2-5}
N_{j_1,j_2}&=\sum_{\gamma\in\fq^*\setminus\{\pm\beta\}}\frac{1}{d}\sum_{i_1=0}^{d-1}\psi^{i_1}((\beta+\gamma)g^{-j_1})\frac{1}{d}\sum_{i_2=0}^{d-1}\psi^{i_2}((\beta-\gamma)g^{-j_2})\nonumber\\
&=\frac{1}{d^2}\sum_{i_1,i_2=0}^{d-1}\sum_{\gamma\in\fq^*\setminus\{\pm\beta\}}\psi^{i_1}((\beta+\gamma)g^{-j_1})\psi^{i_2}((\beta-\gamma)g^{-j_2})\nonumber\\
&=\frac{1}{d^2}\sum_{i_1,i_2=0}^{d-1}S_{i_1,i_2}.
\end{align}
where
$$S_{i_1,i_2}:=\sum_{\gamma\in\fq^*\setminus\{\pm\beta\}}\psi^{i_1}((\beta+\gamma)g^{-j_1})\psi^{i_2}((\beta-\gamma)g^{-j_2}).$$
Now let us compute all sums $S_{i_1,i_2}$. First, one directly has that
\begin{align}\label{c2-6}
S_{0,0}=q-3.
\end{align} If $i_1\ne 0$ and $i_2=0$, one derives that
\begin{align}\label{c2-7}
\left|S_{i_1,0}\right|&=\Big|\sum_{\gamma\in\fq\setminus\{-\beta\}}\psi^{i_1}((\beta+\gamma)g^{-j_1})-\psi^{i_1}(\beta g^{-j_1})-\psi^{i_1}(2\beta g^{-j_1})\Big|\nonumber\\
&\le\Big|\sum_{\gamma\in\fq\setminus\{-\beta\}}\psi^{i_1}((\beta+\gamma)g^{-j_1})\Big|+\left|\psi^{i_1}(\beta g^{-j_1})\right|+\left|\psi^{i_1}(2\beta g^{-j_1})\right|\nonumber\\
&\le 2,
\end{align}
where the last inequality holds since $\sum_{\gamma\in\fq\setminus\{-\beta\}}\psi^{i_1}((\beta+\gamma)g^{-j_1})=0$, $\left|\psi^{i_1}(\beta g^{-j_1})\right|\le 1$ and $\left|\psi^{i_1}(2\beta g^{-j_1})\right|\le 1$.
If $i_1= 0$ and $i_2\ne 0$, similar to the case of $i_1\ne 0$ and $i_2=0$ we have 
\begin{align}\label{c2-8}
\left|S_{0,i_2}\right|\le 2.
\end{align}
Now let $i_1\ne 0$ and $i_2\ne 0$. So if let $\psi(0)=0$ as usual, then we have that
\begin{align*}
|S_{i_1,i_2}|&=\Big|\sum_{\gamma\in\fq^*\setminus\{\pm\beta\}}\psi((\beta+\gamma)^{i_1}(\beta-\gamma)^{i_2}g^{-i_1j_1-i_2j_2})\Big|\\
&=\Big|\sum_{\gamma\in\fq}\psi((\beta+\gamma)^{i_1}(\beta-\gamma)^{i_2}g^{-i_1j_1-i_2j_2})-\psi(\beta^{i_1+i_2}g^{-i_1j_1-i_2j_2})\Big|\\
&\le \Big|\sum_{\gamma\in\fq}\psi((\beta+\gamma)^{i_1}(\beta-\gamma)^{i_2}g^{-i_1j_1-i_2j_2})\Big|+1.
\end{align*}
It follows from the Weil bound (see \cite[Theorem 5.41]{[LN]}) that
\begin{align}\label{c2-9}
|S_{i_1,i_2}|\le 1+q^{1/2}.
\end{align}
Therefore, putting (\ref{c2-5}), (\ref{c2-6}), (\ref{c2-7}), (\ref{c2-8}) and (\ref{c2-9}) together and noting that $d\ge 2$ we have that
\begin{align}\label{c2-10}
|N_{j_1,j_2}-\frac{q-3}{d^2}|&=\frac{1}{d^2}\Big|\sum_{i_1=1}^{d-1}S_{i_1,0}+\sum_{i_2=1}^{d-1}S_{0,i_2}+\sum_{i_1,i_2=1}^{d-1}S_{i_1,i_2}\Big|\nonumber\\
&\le\frac{1}{d^2}\Big(\sum_{i_1=1}^{d-1}|S_{i_1,0}|+\sum_{i_2=1}^{d-1}|S_{0,i_2}|+\sum_{i_1,i_2=1}^{d-1}|S_{i_1,i_2}|\Big)\nonumber\\
&\le \frac{4(d-1)+(d-1)^2(1+q^{1/2})}{d^2}\nonumber\\
&<2+q^{1/2}.
\end{align}
Suppose that $N_{j_1,j_2}=0$. Using (\ref{c2-10}) and noting that $q\ge 5$, one has that 
$$d^2>\left(\frac{1-3/q}{1+2/q^{1/2}}\right)q^{1/2}\ge \left(2-\frac{4}{\sqrt{5}}\right)q^{1/2}=d_0^2,$$
which contradicts the condition $d\le d_0$. So we obtain that $N_{j_1,j_2}>0$. This completes the proof of Lemma \ref{lem2.5}.
\end{proof}

\section{Precise results of $S_{2\ell^m}(a,b)$}
Let the notations $p,\ell,m$ and $q$ be agreed in Section 2. Define the exponential sum $S_{2\ell^m}(a,b)$ over $\fq$ as in (\ref{c2-1-1}). In this section, we shall concentrate our efforts on computing $S_{2\ell^m}(a,b)$ for any $a\in\fth^*$ and $b\in\fq^*$ when $p=3$. In the following, we always assume that $g$ is the fixed primitive element of $\fq$ and $\zeta$ is the primitive cube root of unity. 

Let $J$ be the set of nonnegative integers less than $2\ell^m$, i.e., $J=\{j\in\mathbb{Z}:\ 0\le j\le 2\ell^m-1\}$. Define four subsets of $J$ by $J_1=\{j\in\mathbb{Z}:\ 1\le j\le (\ell-1)\ell^{m-1}\}$, $J_2=\{j\in\mathbb{Z}:\  (\ell-1)\ell^{m-1}< j\le\ell^{m}\}$, $J_3=\{j\in\mathbb{Z}:\  \ell^{m}< j\le 2\ell^{m}-\ell^{m-1}\}$ and $J_4=\{j\in\mathbb{Z}:\  2\ell^{m}-\ell^{m-1}< j\le\ 2\ell^{m}-1\}$. Define three subsets of $J_1$ by $J_1^{(1)}=\{j\in J_1:\ \ell^{m-1}\nmid j\}$, $J_1^{(2)}=\{j\in J_1:\  \ell^{m-1}\mid j\ \text{and}\ 2\nmid j\}$ and $J_1^{(3)}=\{j\in J_1:\  2\ell^{m-1}\mid j\}$. Clearly, $(\{0\},J_1,J_2,J_3,J_4)$ and $(J_1^{(1)},J_1^{(2)},J_1^{(3)})$ are partitions of $J$ and $J_1$, respectively.
One also finds that
$$J_3=\{\ell^m+k\ell^{m-1}-u:\ 1\le k\le \ell-1\ \text{and}\ 0\le u\le\ell^{m-1}-1\}.$$
So $J_3$ can be partitioned into $J_3=J_3^{(1)}\cup J_3^{(2)}\cup J_3^{(3)}$, where
$J_3^{(1)}=\{\ell^m+k\ell^{m-1}-u:\ 1\le k\le \ell-1\ \text{and}\ 0<u\le\ell^{m-1}-1\}$, $J_3^{(2)}=\{\ell^m+k\ell^{m-1}:\ 1\le k\le \ell-1\ \text{and}\ 2\mid k\}$,  and $J_3^{(3)}=\{\ell^m+k\ell^{m-1}:\ 1\le k\le \ell-1\ \text{and}\ 2\nmid k\}$.

Now we are able to state the result of this section.
\begin{thm}\label{thm3.1.1}
Let $b\in\fq^*$ and write $b=g^{-i}$ for some integer $i$ with $0\le i\le q-2$. Let $j$ be the least nonnegative residue of $i$ modulo $2\ell^m$. Define two notation functions $h(x)$ and $H(x)$ by
$$h(x)=-x+\frac{(\sqrt{q}+1)(2\ell^mx-2\ell^m+3\ell-3)}{2\ell^m}$$
and
$$H(x)=-x^{(-1)^m}+\frac{(\sqrt{q}+1)(2\ell^{m-1}x^{(-1)^m}-2\ell^{m-1}+3)}{2\ell^{m-1}}.$$
Then the following statements are true.
\begin{enumerate}[(a)]
\item Let $a\in\fq^*$. Then 
$$S_{2\ell^m}(a,b)=S_{2\ell^m}(a,2b).$$
Furthermore, if $S_{2\ell^m}(a,b)$ is seen formally as a function of $\zeta$ denoted by $S_{2\ell^m}(a,b;\zeta)$, then 
$$S_{2\ell^m}(2a,b)=S_{2\ell^m}(a,b;\zeta^{-1}).$$
\item When $\ell\equiv 1\pmod{3}$, we have that
$$S_{2\ell^m}(1,b)=
\begin{cases}
-1+\frac{(\sqrt{q}+1)(3\ell-3)}{2\ell^m},&\text{if}\ j\in \{0\}\cup J_1^{(1)}\cup J_2\cup J_3^{(1)}\cup J_4,\\
h(\zeta),&\text{if}\ j\in J_1^{(2)}\cup J_3^{(2)},\\
h(\zeta^{-1}),&\text{if}\ j\in J_1^{(3)}\cup J_3^{(3)}.
\end{cases}$$
\item When $\ell\equiv -1\pmod{3}$, we have that
$$S_{2\ell^m}(1,b)=
\begin{cases}
-1+\frac{3(\sqrt{q}+1)}{2\ell^{m-1}},&\text{if}\ j\in  J_1^{(1)}\cup \left(J_2\setminus \{\ell^m\}\right)\cup J_3^{(1)}\cup J_4,\\
H(\zeta),&\text{if}\ j\in J_1^{(3)}\cup \{\ell^m\}\cup J_3^{(3)},\\
H(\zeta^{-1}),&\text{if}\ j\in \{0\}\cup J_1^{(2)}\cup J_3^{(2)}.
\end{cases}$$
\end{enumerate}
Consequently, the values of $S_{2\ell^m}(x,yb)$ are obtained for any $x,y\in\fth^*$.
\end{thm}
\begin{proof}
Let $a,b\in\fq^*$. From Lemma \ref{lem2.1}, we know that
\begin{align}\label{c3-6-1}
S_{2\ell^m}(a,b)=\chi(c_{a,b})\sqrt{q}-\frac{\sqrt{q}+1}{2\ell^m}
S_{2\ell^m}(c_{a,b}),
\end{align}
where $c_{a,b}:=ab^{-\frac{q-1}{2\ell^m}}$. Note that $\phi(2\ell^m)$ is even so that $q\equiv 1\pmod{4}$. It implies that $\frac{q-1}{2\ell^m}$ is even. One then has $c_{a,2b}=a(2b)^{-\frac{q-1}{2\ell^m}}=ab^{-\frac{q-1}{2\ell^m}}=c_{a,b}$ and $c_{2a,b}=2ab^{-\frac{q-1}{2\ell^m}}=-c_{1,1}$. It follows from the formulae (\ref{c3-6-1}) that $S_{2\ell^m}(a,2b)=S_{2\ell^m}(a,b)$ and $S_{2\ell^m}(2a,b)=S_{2\ell^m}(a,b;\zeta^{-1})$. Part (a) is proved.

In the sequel, we are going to compute $S_{2\ell^m}(1,b)$. For the purpose, let $b=g^{-i}$ for some integer $i$ with $0\le i\le q-2$ and let $j$ be the least nonnegative residue of $i$ modulo $2\ell^m$. Recalling that $\xi=g^{\frac{q-1}{2\ell^m}}$, one has $c_{1,b}=b^{-\frac{q-1}{2\ell^m}}=(g^{\frac{q-1}{2\ell^m}})^i=\xi^i=\xi^j$.
Note that $j\in J$ and $J$ has a partition $(\{0\},J_1,J_2,J_3,J_4)$. So we have the following cases.

{\sc Case 1}. $j=0$. In this case, $c_{1,1}=\xi^j=1$. So
$\chi(c_{1,1})=\chi(1)=\zeta^{{\rm Tr}(1)}=\zeta^{\phi(2\ell^m)}$, i.e.,
\begin{align}\label{c3-7-1}
\chi(c_{1,1})=\begin{cases}
1,&\text{if}\ \ell\equiv 1\pmod{3},\\
\zeta^{(-1)^{m-1}},&\text{if}\ \ell\equiv -1\pmod{3}.
\end{cases}
\end{align}
Then substituting (\ref{c3-7-1}) and (\ref{c2-2}) into (\ref{c3-6-1}) yields that
\begin{align*}
S_{2\ell^m}(1,b)=\begin{cases}
-1+\frac{(\sqrt{q}+1)(3\ell-3)}{2\ell^m},&\text{if}\ \ell\equiv 1\pmod{3},\\
-\zeta^{(-1)^{m-1}}+\frac{(\sqrt{q}+1)(2\ell^m\zeta^{(-1)^{m-1}}-2\ell^m+3\ell
)}{2\ell^m},&\text{if}\ \ell\equiv -1\pmod{3}.
\end{cases}
\end{align*}

{\sc Case 2}. $j\in J_1$. By the Euclidean division there exists a unique integer pair $(k,u)$ with $1\le k\le \ell-1$ and $0\le u\le\ell^{m-1}-1$ such that $r=k\ell^{m-1}-u$. Let
$c_{1,b}=\sum_{l=1}^{(\ell-1)\ell^{m-1}}a_l\xi^l$
with each $a_l\in\fth$. Since $c_{1,b}=\xi^j$, then by the uniqueness of the representation, one gets that
\begin{align}\label{c3-8-1}
c_{1,b}^{(u)}=(\underbrace{0,\ldots,0,1,0,\ldots,0}_{\ell-1}),\ \text{and}\ c_{1,b}^{(h)}=(\underbrace{0,0,\ldots,0}_{\ell-1})
\end{align}
for all $h$ with $h\ne u$, where the only one $1$ in the components of $c_{1,b}^{(i_0)}$ is located in the $k$-th position of this vector. Hence, substituting (\ref{c3-8-1}) into Lemma \ref{lem2.2}, one computes that
\begin{align}\label{c3-9-1}
S_{2\ell^m}(c_{1,b})=\begin{cases}
2\ell^m-3\ell+3,&\text{if}\ \ell\equiv 1\pmod{3},\\
2\ell^m-3\ell,&\text{if}\ \ell\equiv -1\pmod{3},\\
\end{cases}
\end{align}
On the other hand, Lemma \ref{lem2.4} gives that
\begin{align}\label{c3-10-1}
\chi(c_{1,b})=\begin{cases}
\zeta^2\ \text{if}\ \ell\equiv 1\pmod{3},\ \text{or}\ \zeta^{(-1)^m} \ \text{otherwise},&j\in J_1^{(3)},\\
\zeta\ \text{if}\ \ell\equiv 1\pmod{3},\ \text{or}\ \zeta^{(-1)^{m-1}} \ \text{otherwise},&j\in J_1^{(2)},\\
1,&j\in J_1^{(1)}.
\end{cases}
\end{align}
Putting Equations (\ref{c3-6-1}), (\ref{c3-9-1}) and (\ref{c3-10-1}) together, we conclude that
\begin{align}\label{c3-11-1}
S_{2\ell^m}(1,b)=\begin{cases}
-\zeta^2+\frac{(\sqrt{q}+1)(2\ell^m\zeta^2-2\ell^m+3\ell-3)}{2\ell^m},&\text{if}\ \ell\equiv 1\pmod{3},\\
-\zeta^{(-1)^m}+\frac{(\sqrt{q}+1)(2\ell^m\zeta^{(-1)^m}-2\ell^m+3\ell)}{2\ell^m},&\text{if}\ \ell\equiv -1\pmod{3}
\end{cases}
\end{align}
when $j\in J_1^{(3)}$,
\begin{align}\label{c3-12-1}
S_{2\ell^m}(1,b)=\begin{cases}
-\zeta+\frac{(\sqrt{q}+1)(2\ell^m\zeta-2\ell^m+3\ell-3)}{2\ell^m},&\text{if}\ \ell\equiv 1\pmod{3},\\
-\zeta^{(-1)^{m-1}}+\frac{(\sqrt{q}+1)(2\ell^m\zeta^{(-1)^{m-1}}-2\ell^m+3\ell)}{2\ell^m},&\text{if}\ \ell\equiv -1\pmod{3}
\end{cases}
\end{align}
when $j\in J_1^{(2)}$, and
\begin{align}\label{c3-13-1}
S_{2\ell^m}(1,b)=\begin{cases}
-1+\frac{(\sqrt{q}+1)(3\ell-3)}{2\ell^m},&\text{if}\ \ell\equiv 1\pmod{3},\\
-1+\frac{3\ell(\sqrt{q}+1)}{2\ell^m},&\text{if}\ \ell\equiv -1\pmod{3}
\end{cases}
\end{align}
when $j\in J_1^{(1)}$.

{\sc Case 3}. $j\in J_2$. Write $j=\ell^{m}-u$ for some integer $u$ with $0\le u\le \ell^{m-1}-1$. By (\ref{c2-3}), we have
$\xi^{\ell^{m-1}(\ell-1)}=\sum_{k=0}^{\ell-2}(-1)^{k-1}\xi^{k\ell^{m-1}}$. It implies that
$$\xi^j=\xi^{\ell^m-u}
=\sum_{k=0}^{\ell-2}(-1)^{k-1}\xi^{(k+1)\ell^{m-1}-u}
=\sum_{k=1}^{\ell-1}(-1)^{k}\xi^{k\ell^{m-1}-u}.$$
So, we have that
\begin{align*}
c_{1,b}^{(u)}=(\underbrace{-1,1,\ldots,-1,1}_{\ell-1}),\ \text{and}\ c_{1,b}^{(h)}=(\underbrace{0,0,\ldots,0}_{\ell-1})
\end{align*}
is an all zero sub-vector for any $h\ne u$. By Lemma \ref{lem2.2} we then easily obtain $S_{2\ell^m}(c_{1,b})$, in fact, which has the same values as (\ref{c3-9-1}). Note that $(\ell-1)\ell^{m-1}<j\le \ell^m$. It follows from Lemma \ref{lem2.4} that $\chi(c_{1,b})=1$ if $j\ne \ell^m$, and $\chi(c_{1,b})=\zeta^{-(\ell-1)\ell^{m-1}}$ if $j=\ell^m$. Therefore, by (\ref{c3-6-1}) we have that
\begin{align}\label{c3-14-1}
S_{2\ell^m}(1,b)=\begin{cases}
-1+\frac{(\sqrt{q}+1)(3\ell-3)}{2\ell^m},&\text{if}\ \ell\equiv 1\pmod{3},\\
-1+\frac{3\ell(\sqrt{q}+1)}{2\ell^m},&\text{if}\ \ell\equiv -1\pmod{3}.
\end{cases}
\end{align}
when $j\ne \ell^m$, and
\begin{align*}
S_{2\ell^m}(1,b)=\begin{cases}
-1+\frac{(\sqrt{q}+1)(3\ell-3)}{2\ell^m},&\text{if}\ \ell\equiv 1\pmod{3},\\
-\zeta^{(-1)^{m}}+\frac{(\sqrt{q}+1)(2\ell^m\zeta^{(-1)^{m}}-2\ell^m+3\ell
)}{2\ell^m},&\text{if}\ \ell\equiv -1\pmod{3}.
\end{cases}
\end{align*}
when $j=\ell^m$.

{\sc Case 4}. $j\in J_3$. First, one partitions $J_3$ into $J_3=\bigcup_{k=1}^{\ell-1}I_k$ with each $I_k=\{i\in\mathbb{Z}:\ \ell^m+(k-1)\ell^{m-1}<i\le \ell^m+k\ell^{m-1}\}$. Then there exists a unique integer $k_0$ with $0\le k_0\le \ell-1$ such that $j\in I_{k_0}$. Write $j=\ell^m+k_0\ell^{m-1}-u_0$ for some integer $u_0$ with $0\le u_0\le \ell^{m-1}-1$. Note that $\xi^{\ell^m}=-1$ since $\xi$ is a primitive $2\ell^m$-th root of unity. It follows that $\xi^j=\xi^{\ell^m}\cdot\xi^{k_0\ell^{m-1}-u_0}=-\xi^{k_0\ell^{m-1}-u_0}$.
Thus,
\begin{align}\label{c3-15-1}
c_{1,b}^{(u_0)}=(\underbrace{0,\ldots,0,-1,0,\ldots,0}_{\ell-1}),\ \text{and}\ c_{1,b}^{(h)}=(\underbrace{0,0,\ldots,0}_{\ell-1})
\end{align}
is an all-zero subvector for any $h$ with $h\ne u_0$, where the only one $-1$ in the components of $c_{1,b}^{(u_0)}$ is located in the $k_0$-th position of this vector. Hence, substituting (\ref{c3-15-1}) into Lemma \ref{lem2.2}, one derives $S_{2\ell^m}(c_{1,b})$, which shares the same value with (\ref{c3-9-1}). Also, by Lemma \ref{lem2.4} one deduces that ${\rm Tr}(\xi^j)$ is equal to $0$ if $j\in J_3^{(1)}$, $-\ell^{m-1}$ if $j\in J_3^{(3)}$, and $\ell^{m-1}$ if $j\in J_3^{(2)}$. Then
\begin{align}\label{c3-16-1}
\chi(c_{1,b})=\begin{cases}
1,&\text{if}\ j\in J_3^{(1)},\\
\zeta^{-\ell^{m-1}},&\text{if}\ j\in J_3^{(3)},\\
\zeta^{\ell^{m-1}},&\text{if}\ j\in J_3^{(2)}.
\end{cases}
\end{align}
Putting the values of $S_{2\ell^m}(c_{1,b})$ and $\chi(c_{1,b})$ into (\ref{c3-6-1}), we obtain that $S_{2\ell^m}(1,b)$ shares the same value with (\ref{c3-13-1}) if $j\in J_3^{(1)}$, shares the same value with (\ref{c3-11-1}) if $j\in J_3^{(3)}$, and shares the same value with (\ref{c3-12-1}) if $j\in J_3^{(2)}$.

{\sc Case 5}. $j\in J_4$. Write $j=2\ell^m-u$ for some integer $u$ with $1\le u\le\ell^{m-1}-1$. Note that $\xi^j=\xi^{\ell^m}\cdot\xi^{\ell^m-u}=-\xi^{\ell^m-u}$. Then from Case 3 we know that
$$\xi^j=\sum_{k=1}^{\ell-1}(-1)^{k-1}\xi^{k\ell^{m-1}-u},$$
which is equivalent to that
\begin{align*}
c_{1,b}^{(u)}=(\underbrace{1,-1,\ldots,1,-1}_{\ell-1}),\ \text{and}\ c_{1,b}^{(h)}=(\underbrace{0,0,\ldots,0}_{\ell-1})
\end{align*}
is an all zero sub-vector for any $h\ne u$. This together with Lemma \ref{lem2.2} gives $S_{2\ell^m}(c_{1,b})$ which has the same value as (\ref{c3-9-1}). Further, by Lemma \ref{lem2.4} one sees that $\chi(c_{1,b})=1$. So we obtain $S_{2\ell^m}(1,b)$, in fact, which shares the same value with (\ref{c3-14-1}) in this case.

Combining all cases above, the desired results follow. This completes the proof of Theorem \ref{thm3.1.1}.
\end{proof}
For any $\alpha\in\fth$ and $\beta\in\fq$, define a key notation $w(\alpha,\beta)$ by
\begin{align}\label{c3-12}
w(\alpha,\beta):=\zeta^{-\alpha}S_{2\ell^m}(1,\beta)+\zeta^{\alpha}S_{2\ell^m}(2,2\beta). 
\end{align}
Let ${\rm Ind}(\cdot)$ be the index of a nonzero element of $\fq$ with respect to the primitive element $g$. Let $J'=\{0\}\cup J_1^{(1)}\cup J_2\cup J_3^{(1)}\cup J_4$ and $J''=J_1^{(1)}\cup (J_2\{\ell^m\})\cup J_3^{(1)}\cup J_4$. For $\alpha\in\fth\setminus\{0\}$ and $\beta\in\fq\setminus\{0\}$, one can write $\alpha=(-1)^{T_{\alpha}}$ with $T_{\alpha}\in\{0,1\}$ and let $j_{\beta}$ be the least nonnegative residue of $-{\rm Ind}(\beta)$ modulo $2\ell^m$. Then by Theorem \ref{thm3.1.1} we immediately obtain the values of $w(\alpha,\beta)$ as follows.  
\begin{cor}\label{cor3.2}
For any $\alpha\in\fth$ and $\beta\in\fq$, we have $w(\alpha,\beta)=w(\alpha,2\beta)$. Further, 
\begin{enumerate}[(a)]
\item when $\ell\equiv 1\pmod{3}$, $w(\alpha,\beta)$ is equal to $\frac{(2\ell^m-3\ell+3)(q-1)}{\ell^m}$ if $\alpha=\beta=0$, $\frac{(2\ell^m-3\ell+3)(1-q)}{2\ell^m}$ if $\alpha\ne 0$ and $\beta=0$, $\frac{(3\ell-3)(\sqrt{q}+1)}{\ell^m}-2$ if $\alpha=0$ and $\beta\ne 0$ with $j_{\beta}\in J'$, $\frac{3(-\ell^m+\ell-1)(\sqrt{q}+1)}{\ell^m}+1$ if $\alpha=0$ and $\beta\ne 0$ with $j_{\beta}\in J\setminus J'$, $\frac{(3-3\ell)(\sqrt{q}+1)}{2\ell^m}+1$ if $\alpha\ne 0$ and $\beta\ne 0$ with $(T_{\alpha},j_{\beta})$ belonging to one of the sets $\{0,1\}\times J'$, $\{0\}\times (J_1^{(3)}\cup J_3^{(3)})$ and $\{1\}\times (J_1^{(2)}\cup J_3^{(2)})$, and $\frac{(3-3\ell)(\sqrt{q}+1)}{2\ell^m}+3\sqrt{q}+1$ if $\alpha\ne 0$ and $\beta\ne 0$ with $(T_{\alpha},j_{\beta})$ belonging to one of the sets $\{1\}\times (J_1^{(3)}\cup J_3^{(3)})$ and $\{0\}\times (J_1^{(2)}\cup J_3^{(2)})$; and
\item  when $\ell\equiv -1\pmod{3}$, $w(\alpha,\beta)$ is equal to $\frac{(2\ell^{m-1}-3)(q-1)}{\ell^{m-1}}$ if $\alpha=\beta=0$, $\frac{(-2\ell^{m-1}-3)(q-1)}{2\ell^{m-1}}$ if $\alpha\ne 0$ and $\beta=0$, $\frac{3(\sqrt{q}+1)}{\ell^{m-1}}-2$ if $\alpha=0$ and $\beta\ne 0$ with $j_{\beta}\in J''$, $\frac{(-3\ell^{m-1}+3)(\sqrt{q}+1)}{\ell^{m-1}}+1$ if $\alpha=0$ and $\beta\ne 0$ with $j_{\beta}\in J\setminus J'$, $\frac{-3(\sqrt{q}+1)}{2\ell^{m-1}}+2$ if $\alpha\ne 0$ and $\beta\ne 0$ with $(m+T_{\alpha},j_{\beta})$ being in one of the sets $\mathbb{Z}\times J''$, $(2\mathbb{Z})\times (\{0\}\cup J_1^{(2)}\cup J_3^{(2)})$ and $(2\mathbb{Z}+1)\times (\{\ell^m\}\cup J_1^{(3)}\cup J_3^{(3)})$, and $\frac{(6\ell^{m-1}-3)(\sqrt{q}+1)}{2\ell^{m-1}}-4$ if $\alpha\ne 0$ and $\beta\ne 0$ with $(T_{\alpha},j_{\beta})$ belonging to one of the sets $(m+T_{\alpha},j_{\beta})$ is in $(2\mathbb{Z}+1)\times (\{0\}\cup J_1^{(2)}\cup J_3^{(2)})$ or in $(2\mathbb{Z})\times (\{\ell^m\}\cup J_1^{(3)}\cup J_3^{(3)})$. Here $2\mathbb{Z}$ represents the set of even integers and $2\mathbb{Z}+1$ the set of odd integers.
\end{enumerate}
\end{cor}

\section{A ternary code and its weight distribution}
In this section, we construct a ternary code and determine its Hamming weight distribution. Let $m$ be a positive integer. Let $\ell$ be an odd prime such that $3$ is a primitive root modulo $2\ell^m$. Let $q=3^{e}$ with $e=\phi(2\ell^m)$, and Tr be the trace from $\mathbb{F}_q$ onto $\mathbb{F}_3$. For any given $\alpha\in\fth$ and $\beta\in\fq$, we shall investigate the Hamming weight distribution of the linear code $\mathcal{C}_{\alpha,\beta}$ over $\fth$ defined by 
\begin{align*}
\mathcal{C}_{\alpha,\beta}:=\{({\rm Tr}(d_1x),{\rm Tr}(d_2x),\ldots,{\rm Tr}(d_nx)):\ x\in\mathbb{F}_q\}.
\end{align*}
with the defining set \begin{align}\label{c3-1}
D=\{x\in\mathbb{F}_q^*: \Tr(x^{\frac{q-1}{2\ell^m}}+\beta x)=\alpha\},
\end{align}
where $n=\#D$. Let $\mathcal{C}_{\alpha,\beta}^{\perp}$ be the dual code of $\mathcal{C}_{\alpha,\beta}$, i.e.,
$$\mathcal{C}_{\alpha,\beta}^{\perp}=\{x\in\fth^n:\ \langle x,c\rangle=0,\ \forall c\in\mathcal{C}_{\alpha,\beta}\},$$
where $\langle x,c\rangle$ denotes the Euclidean inner product of $x$ and $c$.
First, let us determine the minimum distance of $\mathcal{C}_{\alpha,\beta}^{\perp}$ and the number of codewords of $\mathcal{C}_{\alpha,\beta}^{\perp}$ with the minimum Hamming weight.
\begin{prop}\label{prop4.1}
For any $\alpha\in\fth$ and $\beta\in\fq^*$, we have that $\mathcal{C}_{\alpha,\beta}^{\perp}$ is of minimum distance $2$ except for the case of $\ell\equiv-1\pmod{3}$, $m=1$, $\alpha=0$ and $\beta\ne 0$ with $j_{\beta}\in J\setminus J''$ where $\mathcal{C}_{\alpha,\beta}^{\perp}$ is of distance at least $3$. Moreover, if let $A_2^{\perp}$ be the number of codewords of $\mathcal{C}^{\perp}$ with weight $2$, then
\begin{enumerate}[(a)]
\item when $\ell\equiv 1\pmod{3}$, $A_2^{\perp}$ is equal to $\frac{(\ell^m-\ell+1)q}{3\ell^m}+\frac{(-2\ell-2\ell)\sqrt{q}}{3\ell^m}-\frac{(\ell^m-\ell+1)q}{\ell^m}$ if $\alpha=0$ and $\beta\ne 0$ with $j_{\beta}\in J'$, $\frac{(\ell^m-\ell+1)q}{3\ell^m}+\frac{(-2\ell^m+2\ell-2)\sqrt{q}}{3\ell^m}+\frac{-\ell^m+\ell-1}{\ell^m}$ if $\alpha=0$ and $\beta\ne 0$ with $j_{\beta}\in J\setminus J'$, $\frac{(\ell-1)q}{6\ell^m}+\frac{(2\ell^m-\ell+1)\sqrt{q}}{3\ell^m}+\frac{1-\ell}{2\ell^m}$ if $\alpha\ne 0$ and $\beta\ne 0$ with $(T_{\alpha},j_{\beta})$ belonging to one of the sets $\{1\}\times (J_1^{(3)}\cup J_3^{(3)})$ and $\{0\}\times (J_1^{(2)}\cup J_3^{(2)})$, and $\frac{(\ell-1)q}{6\ell^m}+\frac{(1-\ell)\sqrt{q}}{3\ell^m}+\frac{1-\ell}{2\ell^m}$ if $\alpha\ne 0$ and $\beta\ne 0$ with $(T_{\alpha},j_{\beta})$ belonging to one of the sets $\{0,1\}\times J'$, $\{0\}\times (J_1^{(3)}\cup J_3^{(3)})$ and $\{1\}\times (J_1^{(2)}\cup J_3^{(2)})$; and
\item  when $\ell\equiv -1\pmod{3}$, $A_2^{\perp}$ is equal to $\frac{(\ell^{m-1}-1)q}{3\ell^{m-1}}+\frac{2\sqrt{q}}{3\ell^{m-1}}+\frac{1-\ell^{m-1}}{\ell^{m-1}}$ if $\alpha=0$ and $\beta\ne 0$ with $j_{\beta}\in J''$, $\frac{(\ell^{m-1}-1)q}{3\ell^{m-1}}+\frac{(2-2\ell^{m-1})\sqrt{q}}{3\ell^{m-1}}+\frac{1-\ell^{m-1}}{\ell^{m-1}}$ if $\alpha=0$ and $\beta\ne 0$ with $j_{\beta}\in J\setminus J'$, $\frac{q}{6\ell^{m-1}}-\frac{\sqrt{q}}{3\ell^{m-1}}-\frac{1}{2\ell^{m-1}}$ if $\alpha\ne 0$ and $\beta\ne 0$ with $(m+T_{\alpha},j_{\beta})$ being in one of the sets $\mathbb{Z}\times J''$, $(2\mathbb{Z})\times (\{0\}\cup J_1^{(2)}\cup J_3^{(2)})$ and $(2\mathbb{Z}+1)\times (\{\ell^m\}\cup J_1^{(3)}\cup J_3^{(3)})$, and $\frac{q}{6\ell^{m-1}}+\frac{(2\ell^{m-1}-1)\sqrt{q}}{3\ell^{m-1}}-\frac{1}{2\ell^{m-1}}$ if $\alpha\ne 0$ and $\beta\ne 0$ with $(T_{\alpha},j_{\beta})$ belonging to one of the sets $(m+T_{\alpha},j_{\beta})$ is in $(2\mathbb{Z}+1)\times (\{0\}\cup J_1^{(2)}\cup J_3^{(2)})$ or in $(2\mathbb{Z})\times (\{\ell^m\}\cup J_1^{(3)}\cup J_3^{(3)})$. Here $2\mathbb{Z}$ represents the set of even integers and $2\mathbb{Z}+1$ the set of odd integers.
\end{enumerate}
\end{prop}
\begin{proof}
Let $n$ be the size of the defining set $D$, and write $D=\{d_1,\ldots,d_n\}$. The explicit value of $n$ will be derived in Theorem \ref{thm4.2}, from which one sees that $n>0$ when $\beta\ne 0$. It then follows that the minimum distance of $\mathcal{C}_{\alpha,\beta}^{\perp}$ is greater than one. Now let us evaluate $A_2^{\perp}$, the number of codewords in $\mathcal{C}_{\alpha,\beta}^{\perp}$ of weight two. For any codewords $c$ in $\mathcal{C}_{\alpha,\beta}^{\perp}$ of weight two, one lets $c=(c_1,\ldots,c_n)$ with $c_i, c_j\ne 0$ for some two distinct subscripts $i, j$, and $c_k=0$ for any $k\ne i, j$. It implies that $\langle c,c'\rangle=0$ for any $c'\in\mathcal{C}_{\alpha,\beta}$, that is,
$$c_i{\rm Tr}(d_ix)+c_j{\rm Tr}(d_jx)=0,\ \forall x\in\fq.$$
This means that $c_id_i+c_jd_j=0$. Note that $d_i\ne d_j$ and $c_i,c_j\in\fth^*$. It follows that $c_i=c_j$ and $d_j=-d_i$. So there is a one to one relation between codewords in $\mathcal{C}_{\alpha,\beta}^{\perp}$ of weight two and the pairs $(d_i,d_j)$ with $d_i,d_j\in D$ and $d_i+d_j=0$. This gives that 
$$A_2^{\perp}=\#\{(d,d')\in D^2:\ d+d'=0\}.$$
By definition of $D$ and noting that $\frac{q-1}{2\ell^m}$ is even and $\beta\ne 0$, we have that
$$A_2^{\perp}=\#\{x\in \fq^*:\ {\rm Tr}(x^{\frac{q-1}{2\ell^m}})=\alpha\ \text{and}\ {\rm Tr}(\beta x)=0\}.$$
It follows that
\begin{align*}
A_2^{\perp}&=1/9\sum_{x\in\fq^*}\sum_{y\in\fth}\zeta^{y({\rm Tr}(x^{\frac{q-1}{2\ell^m}})-\alpha)}\sum_{z\in\fth}\zeta^{z{\rm Tr}(\beta x)}\\
&=1/9\sum_{x\in\fq^*}\Big(\sum_{y\in\fth^*}\zeta^{y{\rm Tr}(x^{\frac{q-1}{2\ell^m}})-y\alpha}+1\Big)\Big(\sum_{z\in\fth^*}\zeta^{z{\rm Tr}(\beta x)}+1\Big)\\
&=1/9\Big(q-3+\sum_{z\in\fth^*}\sum_{x\in\fq}\chi(zbx)+
\sum_{y\in\fth^*}\zeta^{-y\alpha}\sum_{x\in\fq^*}\chi(yx^{\frac{q-1}{2\ell^m}})\\
&\quad\quad\quad +\sum_{y\in\fth^*}\zeta^{-y\alpha}\sum_{z\in\fth^*}\sum_{x\in\fq^*}\chi(yx^{\frac{q-1}{2\ell^m}}+z\beta x)\Big)\\
&=1/9\Big(q-3+\sum_{y\in\fth^*}\zeta^{-y\alpha}S_{2\ell^m}(y,0)+\sum_{y\in\fth^*}\zeta^{-y\alpha}\sum_{z\in\fth^*}S_{2\ell^m}(y,z\beta)\Big).
\end{align*}
It is to say
\begin{align}\label{c4-2}
A_2^{\perp}=1/9(q-3+w(\alpha,\beta)+w(\alpha,2\beta)+w(\alpha,0)),
\end{align}
where $w(\ ,\ )$ is defined as (\ref{c3-12}). Employing Corollary \ref{cor3.2} into (\ref{c4-2}), we derive the values of $A_2^{\perp}$. From the values of $A_2^{\perp}$, one sees that $A_2^{\perp}>0$ except for the case $\ell\equiv -1\pmod{3}$, $m=1$, $\alpha=0$ and $\beta\ne0$ with $j_{\beta}\in J\setminus J''$ where $A_2^{\perp}=0$. So, the proof of Proposition \ref{prop4.1} is done.
\end{proof}

Now we can report the main result of this section.
\begin{thm}\label{thm4.2}
The following statements regarding the Hamming weight distribution of $\mathcal{C}$ are true.
\begin{enumerate}[(a)]
\item Let $\beta=0$. Then $\mathcal{C}$ is a two-weight code, except for the case $\alpha=0$, $m=1$ and $\ell\equiv -1\pmod{3}$ where $\mathcal{C}$ is an empty set. More precisely, 
\begin{enumerate}[(i)]
\item $\mathcal{C}$ is a $[\frac{(\ell^m-\ell+1)(q-1)}{\ell^m},(\ell-1)\ell^{m-1}]$ code having two weights $w_1=\frac{2(\ell^m-\ell+1)q}{3\ell^m}-\frac{2(\ell-1)\sqrt{q}}{3\ell^m},\ w_2=\frac{2(\ell^m-\ell+1)q}{3\ell^m}+\frac{2(\ell^m-\ell+1)\sqrt{q}}{3\ell^m}$ with $A_{w_1}=\frac{(\ell^m-\ell+1)(q-1)}{\ell^m}, \ A_{w_2}=\frac{(\ell-1)(q-1)}{\ell^m}$ if $\alpha=0$ and $\ell\equiv 1\pmod{3}$;
\item $\mathcal{C}$ is a $[\frac{(\ell-1)(q-1)}{2\ell^m},(\ell-1)\ell^{m-1}]$ code having two weights $w_1=\frac{(\ell-1)q}{3\ell^m}+\frac{(\ell-1)\sqrt{q}}{3\ell^m},\ w_2=\frac{(\ell-1)q}{3\ell^m}-\frac{(2\ell^m-\ell+1)\sqrt{q}}{3\ell^m}$ with $A_{w_1}=\frac{(2\ell^m-\ell+1)(q-1)}{\ell^m}, \ A_{w_2}=\frac{(\ell-1)(q-1)}{\ell^m}$ if $\alpha\ne 0$ and $\ell\equiv 1\pmod{3}$;
\item $\mathcal{C}$ is a $[\frac{(\ell^{m-1}-1)(q-1)}{\ell^{m-1}},(\ell-1)\ell^{m-1}]$ code having two weights $w_1=\frac{2(\ell^{m-1}-1)q}{3\ell^{m-1}}+\frac{2(\ell^{m-1}-1)\sqrt{q}}{3\ell^{m-1}},\ w_2=\frac{2(\ell^{m-1}-1)q}{3\ell^{m-1}}-\frac{2\sqrt{q}}{3\ell^{m-1}}$ with $A_{w_1}=\frac{q-1}{\ell^{m-1}}, \ A_{w_2}=\frac{(\ell^{m-1}-1)(q-1)}{\ell^{m-1}}$ if $\alpha=0$, $\ell\equiv -1\pmod{3}$ and $m\ge 2$;
\item $\mathcal{C}$ is a $[\frac{q-1}{2\ell^{m-1}},(\ell-1)\ell^{m-1}]$ code having two weights $w_1=\frac{q}{3\ell^{m-1}}+\frac{\sqrt{q}}{3\ell^{m-1}},\ w_2=\frac{q}{3\ell^{m-1}}-\frac{(2\ell^{m-1}-1)\sqrt{q}}{3\ell^{m-1}}$ with $A_{w_1}=\frac{(2\ell^{m-1}-1)(q-1)}{2\ell^{m-1}}, \ A_{w_2}=\frac{q-1}{2\ell^{m-1}}$ if $\alpha\ne 0$ and $\ell\equiv -1\pmod{3}$;
\end{enumerate}
\item Let $\beta\ne 0$ and $\alpha=0$. Then $\mathcal{C}$ is a four-weight code, except for the case $m=1$ and $\ell\equiv -1\pmod{3}$ where $\mathcal{C}$ is a $[\frac{q}{3},(\ell-1)\ell^{m-1}]$ code having only two nonzero weights $w_1=\frac{q}{3}$ and $w_2=\frac{2q}{9}$ with enumerates $A_{w_1}=2$ and $A_{w_2}=q-3$. Furthermore, 
\begin{enumerate}[(i)]
\item $\mathcal{C}$ is a $[\frac{q}{3}+\frac{(\ell-1)\sqrt{q}}{\ell^m}+\frac{-\ell^m+p-1}{\ell^m},(\ell-1)\ell^{m-1}]$ code having four weights $w_1=\frac{(\ell-1)(q+\sqrt{q})}{3\ell^m}$, $w_2=\frac{2q}{9}$, $w_3=\frac{2q}{9}+\frac{\sqrt{q}}{3}$, $w_4=\frac{2q}{9}+\frac{2\sqrt{q}}{3}$ with $A_{w_1}=2$, $A_{w_2}=\frac{(\ell^m-\ell+1)^2q}{\ell^{2m}}-\frac{(\ell-1)(\ell^m-2\ell+2)\sqrt{q}}{\ell^{2m}}+\frac{(\ell-1)(\ell^m+\ell-1)}{\ell^{2m}}-3$, $A_{w_3}=\frac{2(\ell-1)(\ell^m-\ell+1)q}{\ell^{2m}}+\frac{2(\ell-1)(\ell^m-2\ell+2)\sqrt{q}}{\ell^{2m}}-\frac{2(\ell-1)^2}{\ell^{2m}}$ and $A_{w_4}=\frac{(\ell-1)^2q}{\ell^{2m}}+\frac{(\ell-1)(-\ell^m+2\ell-2)\sqrt{q}}{\ell^{2m}}-\frac{(\ell-1)(-p^m+p-1)}{\ell^{2m}}$ if $\ell\equiv 1\pmod{3}$, $\alpha=0$ and $\beta\ne 0$ with $j_{\beta}\in J'$;
\item $\mathcal{C}$ is a $[\frac{q}{3}+\frac{(-\ell^m+\ell-1)(\sqrt{q}+1)}{\ell^m},(\ell-1)\ell^{m-1}]$ code having four weights $w_1=\frac{(\ell-1)(q+\sqrt{q})}{3\ell^m}-\frac{\sqrt{q}}{3}$, $w_2=\frac{2q}{9}$, $w_3=\frac{2q}{9}-\frac{\sqrt{q}}{3}$, $w_4=\frac{2q}{9}-\frac{2\sqrt{q}}{3}$ with $A_{w_1}=2$, $A_{w_2}=\frac{(\ell-1)^2q}{\ell^{2m}}+\frac{(\ell-1)(-3\ell^m+2\ell-2)\sqrt{q}}{\ell^{2m}}+\sqrt{q}+\frac{(\ell-1)(-3\ell^m+\ell-1)}{\ell^{2m}}-1$, $A_{w_3}=\frac{2(\ell-1)(\ell^m-\ell+1)q}{\ell^{2m}}+\frac{2(\ell-1)(3\ell^m-2\ell+2)\sqrt{q}}{\ell^{2m}}-2\sqrt{q}+\frac{2(\ell-1)(2\ell^m-\ell+1)}{\ell^{2m}}-2$ and $A_{w_4}=\frac{(\ell-1)(-2\ell^m+\ell-1)q}{\ell^{2m}}+q+\frac{(\ell-1)(-3\ell^m+2\ell-2)\sqrt{q}}{\ell^{2m}}+\sqrt{q}+\frac{(\ell-1)(-p^m+p-1)}{\ell^{2m}}$ if $\ell\equiv 1\pmod{3}$, $\alpha=0$ and $\beta\ne 0$ with $j_{\beta}\in J\setminus J'$;
\item $\mathcal{C}$ is a $[\frac{q}{3}+\frac{\sqrt{q}}{\ell^{m-1}}+\frac{1-\ell^{m-1}}{p^{m-1}},(\ell-1)\ell^{m-1}]$ code having four weights $w_1=\frac{q+\sqrt{q}}{3\ell^{m-1}}$, $w_2=\frac{2q}{9}$, $w_3=\frac{2q}{9}+\frac{\sqrt{q}}{3}$, $w_4=\frac{2q}{9}+\frac{2\sqrt{q}}{3}$ with $A_{w_1}=2$, $A_{w_2}=\frac{(\ell^{m-1}-1)^2q}{\ell^{2m-2}}+\frac{(2-\ell^{m-1})\sqrt{q}}{\ell^{2m-2}}+\frac{\ell^{m-1}+1}{\ell^{2m-2}}$, $A_{w_3}=\frac{2(\ell^{m-1}-1)q}{\ell^{2m-2}}+\frac{2(\ell^{m-1}-1)\sqrt{q}}{\ell^{2m-2}}-\frac{2}{\ell^{2m-2}}$ and $A_{w_4}=\frac{q}{\ell^{2m-2}}+\frac{(2-\ell^{m-1})\sqrt{q}}{\ell^{2m-2}}+\frac{1-\ell^{m-1}}{\ell^{2m-2}}$ if $m\ge 2$, $\ell\equiv -1\pmod{3}$, $\alpha=0$ and $\beta\ne 0$ with $j_{\beta}\in J''$;
\item $\mathcal{C}$ is a $[\frac{q}{3}+\frac{(1-\ell^{m-1})(\sqrt{q}+1)}{\ell^{m-1}},(\ell-1)\ell^{m-1}]$ code having four weights $w_1=\frac{q}{3\ell^{m-1}}+\frac{(1-\ell^{m-1})\sqrt{q}}{3\ell^{m-1}}$, $w_2=\frac{2q}{9}$, $w_3=\frac{2q}{9}-\frac{\sqrt{q}}{3}$, $w_4=\frac{2q}{9}-\frac{2\sqrt{q}}{3}$ with $A_{w_1}=2$, $A_{w_2}=\frac{q}{\ell^{2m-2}}+\frac{(2-3\ell^{m-1})\sqrt{q}}{\ell^{2m-2}}+\sqrt{q}+\frac{1-3\ell^{m-1}}{\ell^{2m-2}}-1$, $A_{w_3}=\frac{2(\ell^{m-1}-1)q}{\ell^{2m-2}}+\frac{2(3\ell^{m-1}-2)\sqrt{q}}{\ell^{2m-2}}-2\sqrt{q}+\frac{2(2\ell^{m-1}-1)}{\ell^{2m-2}}-2$ and $A_{w_4}=\frac{(1-2\ell^{m-1})q}{\ell^{2m-2}}+q-\frac{2(\ell^{m-1}-1)\sqrt{q}}{\ell^{2m-2}}+\sqrt{q}+\frac{1-\ell^{m-1}}{\ell^{2m-2}}$ if $m\ge 2$, $\ell\equiv -1\pmod{3}$, $\alpha=0$ and $\beta\ne 0$ with $j_{\beta}\in J\setminus J''$;
\end{enumerate}
\item Let $\beta\ne 0$ and $\alpha\ne 0$. Then $\mathcal{C}$ is a four-weight code, except for the two cases $m=1$, $\ell=5$, $T_{\alpha}=0$ and $j_{\beta}\in\{0\}\cup J_1^{(2)}\cup J_3^{(2)}$, or $m=1$, $\ell=5$, $T_{\alpha}=1$ and $j_{\beta}\in\{5\}\cup J_1^{(3)}\cup J_3^{(3)}$ where $\mathcal{C}$ is a $[22,4]$ code having only three nonzero weights $w_1=12$, $w_2=18$ and $w_3=15$ with enumerates $A_{w_1}=22$, $A_{w_2}=18$ and $A_{w_3}=40$. Precisely, we have that 
\begin{enumerate}[(i)]
\item $\mathcal{C}$ is a $[\frac{q}{3}+\frac{(1-\ell)(\sqrt{q}+1)}{2\ell^m},(\ell-1)\ell^{m-1}]$ code having four weights $w_1=\frac{(2\ell^m-\ell+1)q}{6\ell^m}+\frac{(1-\ell)\sqrt{q}}{6\ell^m}$, $w_2=\frac{2q}{9}$, $w_3=\frac{2q}{9}-\frac{\sqrt{q}}{3}$, $w_4=\frac{2q}{9}-\frac{2\sqrt{q}}{3}$ with $A_{w_1}=2$, $A_{w_2}=\frac{(\ell-1)(-4\ell^m+\ell-1)q}{4\ell^{2m}}+q+\frac{(\ell-1)(-\ell^m+\ell-1)\sqrt{q}}{2\ell^{2m}}+\frac{(\ell-1)(2\ell^m+\ell-1)}{4\ell^{2m}}-3$, $A_{w_3}=\frac{(\ell-1)(2\ell^m-\ell+1)q}{2\ell^{2m}}+\frac{(\ell-1)(\ell^m-\ell+1)\sqrt{q}}{\ell^{2m}}-\frac{(\ell-1)^2}{2\ell^{2m}}$ and $A_{w_4}=\frac{(\ell-1)^2q}{4\ell^{2m}}+\frac{(\ell-1)(-\ell^m+\ell-1)\sqrt{q}}{2\ell^{2m}}+\frac{(\ell-1)(-2p^m+p-1)}{4\ell^{2m}}$ if $\ell\equiv 1\pmod{3}$, $\alpha\ne 0$ and $\beta\ne 0$ with $(T_{\alpha},j_{\beta})$ belonging to one of the sets $\{0,1\}\times J'$, $\{0\}\times (J_1^{(3)}\cup J_3^{(3)})$ and $\{1\}\times (J_1^{(2)}\cup J_3^{(2)})$;
\item $\mathcal{C}$ is a $[\frac{q+3\sqrt{q}}{3}+\frac{(1-\ell)(\sqrt{q}+1)}{2\ell^m},(\ell-1)\ell^{m-1}]$ code having four weights $w_1=\frac{(2\ell^m-\ell+1)(q+\sqrt{q})}{6\ell^m}$, $w_2=\frac{2q}{9}$, $w_3=\frac{2q}{9}+\frac{\sqrt{q}}{3}$, $w_4=\frac{2q}{9}+\frac{2\sqrt{q}}{3}$ with $A_{w_1}=2$, $A_{w_2}=\frac{(\ell-1)^2q}{4\ell^{2m}}+\frac{(\ell-1)(-3\ell^m+\ell-1)\sqrt{q}}{2\ell^{2m}}+\sqrt{q}+\frac{(\ell-1)(-6\ell^m+\ell-1)}{4\ell^{2m}}$, $A_{w_3}=\frac{(\ell-1)(2\ell^m-\ell+1)q}{2\ell^{2m}}+\frac{(\ell-1)(3\ell^m-\ell+1)\sqrt{q}}{\ell^{2m}}-2\sqrt{q}+\frac{(\ell-1)(4\ell^m-p+1)}{2\ell^{2m}}-2$ and $A_{w_4}=\frac{(\ell-1)(p-1-4\ell^m)q}{4\ell^{2m}}+q+\frac{(\ell-1)(-3\ell^m+\ell-1)\sqrt{q}}{2\ell^{2m}}+\sqrt{q}+\frac{(\ell-1)(p-1-2p^m)}{4\ell^{2m}}$ if $\ell\equiv 1\pmod{3}$$\alpha\ne 0$ and $\beta\ne 0$ with $(T_{\alpha},j_{\beta})$ belonging to one of the sets $\{1\}\times (J_1^{(3)}\cup J_3^{(3)})$ and $\{0\}\times (J_1^{(2)}\cup J_3^{(2)})$;
\item $\mathcal{C}$ is a $[\frac{q}{3}-\frac{(\sqrt{q}+1)}{2\ell^{m-1}},(\ell-1)\ell^{m-1}]$ code having four weights $w_1=\frac{(2\ell^{m-1}-1)q-\sqrt{q}}{6\ell^{m-1}}$, $w_2=\frac{2q}{9}$, $w_3=\frac{2q}{9}-\frac{\sqrt{q}}{3}$, $w_4=\frac{2q}{9}-\frac{2\sqrt{q}}{3}$ with $A_{w_1}=2$, $A_{w_2}=\frac{(-4\ell^{m-1}+1)q}{4\ell^{2m-2}}+q+\frac{(\ell^{m-1}-1)\sqrt{q}}{2\ell^{2m-2}}+\frac{2\ell^{m-1}+1}{4\ell^{2m-2}}-3$, $A_{w_3}=\frac{(2\ell^{m-1}-1)q}{2\ell^{2m-2}}+\frac{(\ell^{m-1}-1)\sqrt{q}}{\ell^{2m-2}}-\frac{1}{2\ell^{2m-2}}$ and $A_{w_4}=\frac{q}{4\ell^{2m-2}}+\frac{(1-\ell^{m-1})\sqrt{q}}{2\ell^{2m-2}}+\frac{1-2\ell^{m-1}}{4\ell^{2m-2}}$ if $\ell\equiv -1\pmod{3}$, $\ell^m\ne 5$, $\alpha\ne 0$ and $\beta\ne 0$ with $(m+T_{\alpha},j_{\beta})$ being in one of the sets $\mathbb{Z}\times J''$, $(2\mathbb{Z}+1)\times (\{\ell^m\}\cup J_1^{(3)}\cup J_3^{(3)})$ and $(2\mathbb{Z})\times (\{0\}\cup J_1^{(2)}\cup J_3^{(2)})$;
\item $\mathcal{C}$ is a $[\frac{q+3\sqrt{q}}{3}-\frac{(\sqrt{q}+1)}{2\ell^{m-1}},(\ell-1)\ell^{m-1}]$ code having four weights $w_1=\frac{(2\ell^{m-1}-1)(q+\sqrt{q})}{6\ell^{m-1}}$, $w_2=\frac{2q}{9}$, $w_3=\frac{2q}{9}+\frac{\sqrt{q}}{3}$, $w_4=\frac{2q}{9}+\frac{2\sqrt{q}}{3}$ with $A_{w_1}=2$, $A_{w_2}=\frac{q}{4\ell^{2m-2}}+\frac{(1-3\ell^{m-1})\sqrt{q}}{2\ell^{2m-2}}+\sqrt{q}+\frac{1-6\ell^{m-1}}{4\ell^{2m-2}}-1$, $A_{w_3}=\frac{2(\ell^{m-1}-1)q}{2\ell^{2m-2}}+\frac{(3\ell^{m-1}-1)\sqrt{q}}{\ell^{2m-2}}-2\sqrt{q}+\frac{4\ell^{m-1}-1}{\ell^{2m-2}}-2$ and $A_{w_4}=\frac{(1-4\ell^{m-1})q}{4\ell^{2m-2}}+q+\frac{(1-3\ell^{m-1})\sqrt{q}}{2\ell^{2m-2}}+\sqrt{q}+\frac{1-2\ell^{m-1}}{4\ell^{2m-2}}$ if $\ell\equiv -1\pmod{3}$, $\alpha\ne 0$ and $\beta\ne 0$ with $(m+T_{\alpha},j_{\beta})$ being in $(2\mathbb{Z}+1)\times (\{0\}\cup J_1^{(2)}\cup J_3^{(2)})$ or in $(2\mathbb{Z})\times (\{\ell^m\}\cup J_1^{(3)}\cup J_3^{(3)})$.
\end{enumerate}
\end{enumerate}
\end{thm}
\begin{proof}
First, we compute the length of the code $\mathcal{C}_{\alpha,\beta}$ denoted by $n$, i.e., 
$$n=\#\{x\in\fq^*: \Tr(x^{\frac{q-1}{2\ell^m}}+\beta x)=\alpha\}.$$
For this, let us compute
$n_0:=\#\{x\in\fq: \Tr(x^{\frac{q-1}{2\ell^m}}+\beta x)=\alpha\}$ since $n=n_0- \mathbbm{1}(\alpha=0)$, where $\mathbbm{1}(\alpha=0)$ is $1$ if $\alpha=0$, $0$ otherwise. Note that
$$n_0=\frac{1}{3}\sum_{x\in\fq}\sum_{y\in\fth}\zeta^{y({\rm Tr}(x^{\frac{q-1}{2\ell^m}}+\beta x)-\alpha)}.$$
This gives that
\begin{align}\label{c4-3}
n_0&=\frac{q}{3}+\frac{1}{3}\sum_{y\in\fth^*}\zeta^{-y\alpha}
\Big(\sum_{x\in\fq^*}\chi(yx^{\frac{q-1}{2\ell^m}}+y\beta x)+1\Big)\nonumber\\
&=\frac{q}{3}+\frac{1}{3}\sum_{y\in\fth^*}\zeta^{-y\alpha}
\Big(S_{2\ell^m}(y,y\beta)+1\Big)\nonumber\\
&=\frac{q}{3}+\frac{1}{3}
\left(w(\alpha,\beta)+\zeta^{-\alpha}+\zeta^{\alpha}\right).
\end{align}
Note that $n=n_0- \mathbbm{1}(\alpha=0)$ and $\zeta^{-\alpha}+\zeta^{\alpha}=2$ if $\alpha=0$, $\zeta^{-\alpha}+\zeta^{\alpha}=-1$ otherwise. It follows that
\begin{align}\label{c4-4}
n=\frac{q+w(\alpha,\beta)-1}{3}.
\end{align}
Therefore, by Corollary \ref{cor3.2} we obtain the  length of $\mathcal{C}_{\alpha,\beta}$ as given in this theorem. From this result, one observes that $n>0$ with a exceptional case $\alpha=\beta=0$, $\ell\equiv -1\pmod{3}$ and $m=1$ where it is nonsense. So, in the following we assume that the exceptional case cannot happen.

Next, we are going to determine all nonzero weights of codewords of $\mathcal{C}_{\alpha,\beta}$. Let $wt(\cdot)$ be the Hamming weight of a codeword. We shall figure out $wt(c_{\gamma})$ for any $\gamma\in\fq^*$, where $c_{\gamma}:=(\Tr(\gamma d_1),\ldots,\Tr(\gamma d_n))\in\mathcal{C}_{\alpha,\beta}$. For $\gamma\in\fq^*$ let $N_{\gamma}$ be defined by
$$N_{\gamma}:=\#\{x\in\fq:\ \Tr(x^{\frac{q-1}{2\ell^m}}+\beta x)=\alpha\ \text{and}\ \Tr(\gamma x)=0\}.$$
Then one easily finds that $wt(c_{\gamma})=n_0-N_{\gamma}$.
So, in the following we just need to compute $N_{\gamma}$. By definition, we have that
\begin{align*}
N_{\gamma}&=1/9\sum_{x\in\fq}\sum_{y\in\fth}\zeta^{y\big({\rm Tr}(x^{\frac{q-1}{2\ell^m}}+\beta x)-\alpha\big)}\sum_{z\in\fth}\zeta^{z{\rm Tr}(\gamma x)}\nonumber\\
&=1/9\sum_{x\in\fq}\Big(\sum_{y\in\fth^*}\zeta^{y\big({\rm Tr}(x^{\frac{q-1}{2\ell^m}}+\beta x)-\alpha\big)}+1\Big)\big(\sum_{z\in\fth^*}\zeta^{z{\rm Tr}(\gamma x)}+1\Big)\nonumber\\
&=1/9\Big(q+\sum_{z\in\fth^*}\sum_{x\in\fq}\chi(z\gamma x)+
\sum_{y\in\fth^*}\zeta^{-y\alpha}\sum_{x\in\fq}\chi(yx^{\frac{q-1}{2\ell^m}}+y\beta x)\\
&\quad\quad+\sum_{y\in\fth^*}\zeta^{-y\alpha}\sum_{z\in\fth^*}\sum_{x\in\fq}
\chi(yx^{\frac{q-1}{2\ell^m}}+(y\beta+z\gamma)x)\Big).
\end{align*}
It is reduced to that 
\begin{align*}
N_{\gamma}=1/9\big(q+3\zeta^{-\alpha}+3\zeta^{\alpha}+w(\alpha,\beta)+w(\alpha,\beta+\gamma)+w(\alpha,\beta+2\gamma)\big).
\end{align*}
since $\sum_{x\in\fq}\chi(z\gamma x)=0$ for any $z,\gamma\in\fth^*$. Note that $wt(c_{\gamma})=n_0-N_{\gamma}$. It then follows from (\ref{c4-3}) that
\begin{align}\label{c4-5}
wt(c_{\gamma})=\frac{1}{9}\big(2q+2w(\alpha,\beta)-w(\alpha,\beta+\gamma)-w(\alpha,\beta-\gamma)\big).
\end{align}
From (\ref{c4-5}), the computation of $wt(c_{\gamma})$ can be divided into the following two cases.

$\bullet$ $\beta=0$. In this case,  $wt(c_{\gamma})=\frac{2q}{9}+\frac{2w(\alpha,0)}{9}-\frac{w(\alpha,\gamma)}{9}-\frac{w(\alpha,2\gamma)}{9}=\frac{2q}{9}+\frac{2(w(\alpha,0)-w(\alpha,\gamma))}{9}$. By Corollary \ref{cor3.2} we conclude that $wt(c_{\gamma})$ takes on exactly two values as presented in this theorem when $\gamma$ runs through $\fq^*$.

$\bullet$ $\beta\ne 0$. From (\ref{c4-5}), we need to figure out all possible values $w(\alpha,\beta+\gamma)+w(\alpha,\beta-\gamma)$ takes on once $\alpha$, $\beta$ and $\ell^m$ are given when $\gamma$ runs through $\fq^*$.  If $\gamma=\pm \beta$, then $wt(c_{\gamma})$ equals $\frac{2q}{9}+\frac{w(\alpha,\beta)}{9}-\frac{w(\alpha,0)}{9}$, denoted by $w_1$, which can be derived precisely from Corollary \ref{cor3.2}. Now let $\gamma$ run over $\fq^*\setminus\{\pm\beta\}$. From Corollary \ref{cor3.2}, we know that $w(\alpha,b)$ has two values, like $A$ and $B$, when $b$ runs through $\fq^*$. Without loss of generality, we assume $w(\alpha,\beta)=A$. By Lemma \ref{lem2.5}, we deduce that $w(\alpha,\beta+\gamma)+w(\alpha,\beta-\gamma)$ can take on the three different values $2A$, $2B$ and $A+B$ when $\gamma$ runs over $\fq^*$. It means that all the possible value of $wt(c_{\gamma})$ are $\frac{2q}{9}$, $\frac{2q}{9}+\frac{A-B}{9}$ and $\frac{2q}{9}+\frac{2(A-B)}{9}$, denoted by $w_2$, $w_3$ and $w_4$, respectively. So $\mathcal{C}_{\alpha}$ is a four-weight code with weights $w_1$, $w_2$, $w_3$ and $w_4$, and the number of codewords having weight $w_1$ is exactly $2$, except for three cases $\alpha=0$, $\beta\ne0$, $\ell\equiv -1\pmod{3}$ and $m=1$, or $m=1$, $\ell=5$, $T_{\alpha}=0$ and $j_{\beta}\in\{0\}\cup J_1^{(2)}\cup J_3^{(2)}$, or $m=1$, $\ell=5$, $T_{\alpha}=1$ and $j_{\beta}\in\{5\}\cup J_1^{(3)}\cup J_3^{(3)}$. In the first exceptional case, one observes that $A=B$ so that $\mathcal{C}_{\alpha,\beta}$ is a two-weight code with weights $w_1$ and $w_2$. In the latter two cases, one checks that $w_1=w_4$, and then $\mathcal{C}_{\alpha,\beta}$ is a three-weight code with weights $w_1$, $w_2$ and $w_3$.

Hence, we obtain all nonzero weights of $\mathcal{C}_{\alpha,\beta}$ as given in this theorem. In the sequel, let us work out the weight enumerates as follows.

If $\mathcal{C}_{\alpha,\beta}$ is a two-weight code having weights $w_1$ and $w_2$, and let $A_{w_1}$ and $A_{w_2}$ be their weight enumerates. Then by 
the first two Pless Power Moments \cite{[HCP]}, we have
$$\begin{cases}
A_{w_1}+A_{w_2}=q-1,\\
A_{w_1}w_1+A_{w_2}w_2=\frac{2nq}{3}.
\end{cases}$$
From this, one derive the weight enumerates $A_{w_1}$ and $A_{w_2}$.

If $\mathcal{C}_{\alpha,\beta}$ is a three-weight code having weights $w_1$, $w_2$ and $w_3$, and let $A_{w_1}$, $A_{w_2}$ and $A_{w_3}$ be their weight enumerates. Then by the first three Pless Power Moments \cite{[HCP]}, we have
$$\begin{cases}
A_{w_1}+A_{w_2}+A_{w_2}=q-1,\\
A_{w_1}w_1+A_{w_2}w_2+A_{w_3}w_3=\frac{2nq}{3},\\
A_{w_1}w_1^2+A_{w_2}w_2^2+A_{w_3}w_3^2=\frac{1}{9}q(2n(2n+1) + 2A_2^{\perp}),
\end{cases}$$
where $A_2^{\perp}$ denotes the number of codeword in the dual code with weight two. From this together with the result of $A_2^{\perp}$ in Proposition \ref{prop4.1}, the weight enumerates $A_{w_1}$, $A_{w_2}$ and $A_{w_3}$ can be obtained.

If $\mathcal{C}_{\alpha,\beta}$ is a four-weight code having weights $w_1$, $w_2$, $w_3$ and $w_4$, and let $A_{w_1}$, $A_{w_2}$, $A_{w_3}$ and $A_{w_4}$ be their weight enumerates. Also from the first three Pless Power Moments, we derive that
$$\begin{cases}
A_{w_1}+A_{w_2}+A_{w_2}+A_{w_4}=q-1,\\
A_{w_1}w_1+A_{w_2}w_2+A_{w_3}w_3+A_{w_4}w_4=\frac{2nq}{3},\\
A_{w_1}w_1^2+A_{w_2}w_2^2+A_{w_3}w_3^2++A_{w_4}w_4^2=\frac{1}{9}q(2n(2n+1) + 2A_2^{\perp}).
\end{cases}$$
Note that $A_{w_1}=2$. By solving this linear equations with three unknowns, the weight enumerates $A_{w_1}$, $A_{w_2}$,$A_{w_3}$ and $A_{w_4}$ are settled.

Finally, let us determine the dimension of the code $\mathcal{C}_{\alpha,\beta}$. From the computation above we know that $wt(c)>0$ for any nonzero codeword $c\in\mathcal{C}_{\alpha,\beta}$. It follows that the map 
$$x\mapsto (\Tr(xd_1),\Tr(xd_2),\ldots,\Tr(xd_n))$$
from $\fq$ to $\mathcal{C}_{\alpha,\beta}$ is a bijection. It implies that
$\#\mathcal{C}_{\alpha,\beta}=q$, that is, the dimension of $\mathcal{C}_{\alpha,\beta}$ is $(\ell-1)\ell^{m-1}$ as desired. 

This completes the proof of Theorem \ref{thm4.2}.
\end{proof}
The striking similarity between the Hamming weight distribution of the code $C_{0,0}$ we constructed and that in reference  \cite{[CG]} has piqued our interest on the trace values of $x^{\frac{q-1}{\ell^m}}$ and $x^{\frac{q-1}{2\ell^m}}$ for any $x\in\fq^*$. Define a partition $(F_1,F_2,F_3,F_4,F_4,F_6)$ of $\fq^*$ by 
\begin{align*}
F_1&=\{x\in\fq^*:\ x^{\frac{q-1}{2\ell^m}}=1\},\\
F_2&=\{x\in\fq^*:\ x^{\frac{q-1}{2\ell^m}}\ne 1, x^{\frac{q-1}{\ell^{m}}}=1\ \text{and}\ x^{\frac{q-1}{2\ell^{m-1}}}=1\},\\
F_3&=\{x\in\fq^*:\ x^{\frac{q-1}{2\ell^m}}\ne 1, x^{\frac{q-1}{\ell^{m}}}=1\ \text{and}\ x^{\frac{q-1}{2\ell^{m-1}}}\ne 1\},\\
F_4&=\{x\in\fq^*:\ x^{\frac{q-1}{2\ell^m}}\ne 1, x^{\frac{q-1}{\ell^{m}}}\ne 1,\ x^{\frac{q-1}{2\ell^{m-1}}}=1\  \text{and}\ x^{\frac{q-1}{\ell^{m-1}}}=1\},\\
F_5&=\{x\in\fq^*:\ x^{\frac{q-1}{2\ell^m}}\ne 1, x^{\frac{q-1}{\ell^{m}}}\ne 1,\ x^{\frac{q-1}{2\ell^{m-1}}}\ne1\  \text{and}\ x^{\frac{q-1}{\ell^{m-1}}}=1\},\\
F_6&=\{x\in\fq^*:\ x^{\frac{q-1}{2\ell^m}}\ne 1, x^{\frac{q-1}{\ell^{m}}}\ne 1,\ x^{\frac{q-1}{2\ell^{m-1}}}\ne1\  \text{and}\ x^{\frac{q-1}{\ell^{m-1}}}\ne1\}.
\end{align*}
\begin{prop}\label{prop4.3}
For $x\in\fq^*$ let $\Tr_1$ and $\Tr_2$ be the traces of $x^{\frac{q-1}{\ell^m}}$ and $x^{\frac{q-1}{2\ell^m}}$, respectively. Then we have that
$$\Tr_1=\begin{cases}
0,&\text{if}\ x\in F_1\cup F_2\cup F_3\cup F_6,\\
-1,&\text{if}\ x\in F_4\cup F_5,
\end{cases},\ \Tr_2=\begin{cases}
0,&\text{if}\ x\in F_1\cup F_3\cup F_6,\\
1,&\text{if}\ x\in F_2\cup F_5,\\
-1,&\text{if}\ x\in F_4
\end{cases}$$
if $\ell\equiv 1\pmod{3}$; and
$$\Tr_1=\begin{cases}
(-1)^{m-1},&\text{if}\ x\in F_1\cup F_2\cup F_3,\\
(-1)^m,&\text{if}\ x\in F_4\cup F_5,\\
0,&\text{if}\ x\in F_6,
\end{cases},\ \Tr_2=\begin{cases}
(-1)^{m-1},&\text{if}\ x\in F_1\cup F_5,\\
(-1)^m,&\text{if}\ x\in F_3\cup F_4,\\
0,&\text{if}\ x\in F_2\cup F_6
\end{cases}$$
if $\ell\equiv -1\pmod{3}$.
\end{prop}
\begin{proof}
Let $x\in\fq^*$. It is noted that $3$ is simultaneously a primitive root modulo  $\ell^m$ and $2\ell^m$. By definition of trace, we have that
$$\Tr(x^{\frac{q-1}{\ell^m}})=\sum_{j=0}^{\phi(\ell^m)-1}\big(x^{\frac{q-1}{\ell^m}}\big)^{3^j}=\sum_{k\in\mathbb{Z}_{\ell^m}^*}\big(x^{\frac{q-1}{\ell^m}}\big)^{k}=\sum_{k=0}^{\ell^m-1}\big(x^{\frac{q-1}{\ell^m}}\big)^{k}-\sum_{k=0}^{\ell^{m-1}-1}\big(x^{\frac{q-1}{\ell^{m-1}}}\big)^{k},$$
and
$$\Tr(x^{\frac{q-1}{2\ell^m}})=\sum_{j=0}^{\phi(\ell^m)-1}\big(x^{\frac{q-1}{2\ell^m}}\big)^{3^j}=\sum_{k\in\mathbb{Z}_{2\ell^m}^*}\big(x^{\frac{q-1}{2\ell^m}}\big)^{k}=\sum_{k=0}^{\ell^m-1}\big(x^{\frac{q-1}{2\ell^m}}\big)^{2k+1}-\sum_{k=0}^{\ell^{m-1}-1}\big(x^{\frac{q-1}{2\ell^{m-1}}}\big)^{2k+1}.$$
For convenience, we write 
$\Tr_1:=\Tr(x^{\frac{q-1}{\ell^m}})=\Sigma_1-\Sigma_2$ and $\Tr_2:=\Tr(x^{\frac{q-1}{\ell^m}})=\Sigma_3-\Sigma_4$. Here $\Sigma_1$ and $\Sigma_2$ represent the last two sums of the first equation; $\Sigma_3$ and $\Sigma_4$ represent the last two sums of the second equation. Consequently, we compute the four sums for each case where $x$ is an element of $F_i$ with $1\le i\le 6$, respectively, and subsequently derive the two traces as shown below:\\
$\bullet$ $x\in F_1$. Clearly, $x^{\frac{q-1}{2\ell^m}}=x^{\frac{q-1}{\ell^m}}=x^{\frac{q-1}{2\ell^{m-1}}}=x^{\frac{q-1}{\ell^{m-1}}}=1$. It implies that $\Tr_1=\Tr_2=\ell^m-\ell^{m-1}$.\\
$\bullet$ $x\in F_2$. In this case, we have that $x^{\frac{q-1}{2\ell^m}}=-1$ and $x^{\frac{q-1}{\ell^m}}=x^{\frac{q-1}{2\ell^{m-1}}}=x^{\frac{q-1}{\ell^{m-1}}}=1$. It implies that $\Tr_1=\ell^m-\ell^{m-1}$ and $\Tr_2=-\ell^m-\ell^{m-1}$.\\
$\bullet$ $x\in F_3$. Then $x^{\frac{q-1}{2\ell^m}}=x^{\frac{q-1}{2\ell^{m-1}}}=-1$ and $x^{\frac{q-1}{\ell^m}}=x^{\frac{q-1}{\ell^{m-1}}}=1$. It gives that $\Tr_1=\ell^m-\ell^{m-1}$ and $\Tr_2=-\ell^m+\ell^{m-1}$.\\
$\bullet$ $x\in F_4$. Then one finds that $x^{\frac{q-1}{2\ell^m}}=\zeta_{\ell}$, $x^{\frac{q-1}{\ell^m}}=\zeta_{\ell}'$ for some two primitive $\ell$th roots of unity $\zeta_{\ell}$ and $\zeta_{\ell}'$, and $x^{\frac{q-1}{2\ell^{m-1}}}=x^{\frac{q-1}{\ell^{m-1}}}=1$.  So we have that $\Sigma_1=\frac{1-\zeta_{\ell}'^{\ell^m}}{1-\zeta_{\ell}'}=0$, $\Sigma_2=\Sigma_4=\ell^{m-1}$ and $\Sigma_3=\frac{\zeta_{\ell}(1-\zeta_{\ell}'^{\ell^m})}{1-\zeta_{\ell}'}=0$. This gives that $\Tr_1=\Tr_2=-\ell^m$.\\
$\bullet$ $x\in F_5$. In this case, one deduces that $x^{\frac{q-1}{2\ell^{m-1}}}=-1$, $x^{\frac{q-1}{\ell^{m-1}}}=1$, $x^{\frac{q-1}{2\ell^m}}=\zeta_{2\ell}$,  and $x^{\frac{q-1}{\ell^m}}=\zeta_{\ell}$ for a primitive $2\ell$th root of unity $\zeta_{2\ell}$ and a primitive $\ell$th root of unity $\zeta_{\ell}$. So, $\Sigma_1=\frac{1-\zeta_{\ell}^{\ell^m}}{1-\zeta_{\ell}}=0$, $\Sigma_2=\ell^{m-1}$, $\Sigma_3=\frac{\zeta_{2\ell}(1-\zeta_{\ell}^{\ell^m})}{1-\zeta_{\ell}}=0$ and $\Sigma_4=-\ell^{m-1}$. It follows that $\Tr_1=-\ell^{m-1}$ and $\Tr_2=\ell^{m-1}$.\\
$\bullet$ $x\in F_6$. Then one has that $x^{\frac{q-1}{2\ell^m}}$, $x^{\frac{q-1}{\ell^m}}$, $x^{\frac{q-1}{2\ell^{m-1}}}$ and $x^{\frac{q-1}{\ell^{m-1}}}$ are roots of unity of orders $2\ell^m$, $\ell^m$, $2\ell^{m-1}$, and $\ell^{m-1}$, respectively, and none of them equals 1. It implies that $\Tr_1=\Tr_2=0$.

Combining all the cases above, we arrive at the desired results.
\end{proof}
\begin{rem}
By replacing $2\ell^m$ with $\ell^m$ in the code $\mathcal{C}_{\alpha,\beta}$, we can adapt the techniques in this paper to derive analogous results. Proposition \ref{prop4.3} suggests that some conclusions will remain unchanged, while others may differ.
\end{rem}

\section{The dual codes}
For $\alpha\in\fth$ and $\beta\in\fq$, let $\mathcal{C}_{\alpha,\beta}$ be the code defined as in Section 4, and $\mathcal{C}_{\alpha,\beta}^{\perp}$ be its dual code. Then we have the following results.
\begin{thm}\label{thm5.1}
Let $n$ be the length of $\mathcal{C}_{\alpha,\beta}$. Then $\mathcal{C}_{\alpha,\beta}^{\perp}$ is an $\left[n,n-\ell^{m-1}(\ell-1),2\right]$ code with a exceptional case $\ell\equiv-1\pmod{3}$, $m=1$, $\alpha=0$ and $\beta\ne 0$ with $j_{\beta}\in J\setminus J''$ where $\mathcal{C}_{\alpha,\beta}^{\perp}$ is is an $\left[n,n-\ell^{m-1}(\ell-1),d\right]$ code with $d\ge 3$.
\end{thm}
\begin{proof}
First, observe that the dimensions of the dual code $\mathcal{C}_{\alpha,\beta}^{\perp}$ and the original code $\mathcal{C}_{\alpha,\beta}$ sum to $n$, as stated in Theorem \ref{thm4.2}. Therefore, the dimension of $\mathcal{C}_{\alpha,\beta}^{\perp}$ can be deduced directly from this theorem.

Next, we determine the minimum distance $d$ of $\mathcal{C}_{\alpha,\beta}^{\perp}$. It is evident that $d \geq 2$ due to the surjectivity of the trace map $\Tr$. When $\beta \neq 0$, Proposition \ref{prop4.1} implies that $A_2^{\perp} = 0$ if $\ell \equiv -1 \pmod{3}$, $m = 1$, $\alpha = 0$, and $\beta \neq 0$ with $j_{\beta} \in J \setminus J''$. Otherwise, $A_2^{\perp} > 0$, where $A_2^{\perp}$ denotes the number of codewords in $\mathcal{C}_{\alpha,\beta}^{\perp}$ with weight two. This indicates that $\mathcal{C}_{\alpha,\beta}^{\perp}$ has a minimum distance of 2 except for the aforementioned exceptional case, where $d \geq 3$.

When $\beta = 0$, let $D = \{d_1, d_2, \ldots, d_n\}$. For any element $d_i \in D$, it is clear that $-d_i \in D$ as well. Since $d_i \neq -d_i$, we must have $-d_i = d_j$ for some $j \neq i$. Consequently, the codeword $c = (c_1, c_2, \ldots, c_n)$, where $c_i = c_j = 1$ and $c_k = 0$ for all $k \neq i, j$, belongs to $\mathcal{C}_{\alpha,\beta}^{\perp}$. This implies that $\mathcal{C}_{\alpha,\beta}^{\perp}$ contains a codeword of weight two, and therefore, $d = 2$ when $\beta = 0$.

This completes the proof of Theorem \ref{thm5.1}.
\end{proof}
\begin{defn}\label{defn5.2}
A code $\mathcal{C}$ with parameter $[n,k,d]$ is called \textit{optimal} if no $[n,k,d+1]$ code exists.
\end{defn}
Then by the Hamming bound (or the sphere packing bound) and with some direct checking, the last result of the paper follows immediately.
\begin{cor}\label{cor5.2}
if $\alpha=\beta=0$, then the dual code $\mathcal{C}_{\alpha,\beta}^{\perp}$ is optimal.
\end{cor}
\begin{proof}
Let $n$ and $d$ denote the length and minimum distance of the dual code $\mathcal{C}_{\alpha,\beta}^{\perp}$, respectively. According to the sphere packing bound, we have
$$\big|\mathcal{C}_{\alpha,\beta}^{\perp}\big| \leq \frac{3^n}{\sum_{i=0}^t \binom{n}{i}2^i},$$
where $t$ is the largest integer less than or equal to $(d-1)/2$. Since $d = 2$, it follows that the dual code $\mathcal{C}_{\alpha,\beta}^{\perp}$ is optimal if $n > (q-1)/2$. Consequently, the desired result is a direct consequence of Theorem \ref{thm4.2}.
\end{proof}

\section*{acknowledgment}
The author is particularly grateful to Professor Arne Winterhof at RICAM, Austrian Academy of Sciences, for his invaluable guidance and support during the author's academic visit. His expertise in number theory and his insightful suggestions on the application of character sums to our problem have been instrumental in the completion of this work.



\begin{thebibliography}{99}

\bibitem{[ADHK]} R. Anderson, C. Ding, T. Helleseth and T. Kl{\o}ve, How to build robust shareed control systems, Des. Codes Cryptogr. \textbf{15} (1998), 111-124.

\bibitem{[Car1]} L. Carlitz, \textit{Explicit evaluation of certain exponential sums}, Math. Scand. \textbf{44}(1979), 5-16.

\bibitem{[Car2]} L. Carlitz, \textit{Evaluation of of some exponential sums over a finite field}, Math. Nachr. \textbf{96}(1980), 319-339.

\bibitem{[CG]} K. Cheng and S. Gao, \textit{On binomial Weil sums and a application}, 2024, arXiv:2409.13515.
\bibitem{[Cou]} R. Coulter, \textit{Explicit evaluations of some Weil sums}, Acta. Arith.\textbf{83}(1998), 241-251.

\bibitem{[CK]} A. Calderbank and W. Kantor, \textit{The geometry of two-weight codes}, Bull. London Math. Soc. 18(1986), 97-122.

\bibitem{[Din1]} C. Ding, \textit{Linear codes from some $2$-design}, IEEE Trans. Inf. Theory \textbf{60} (2015), 3265-3275.

\bibitem{[Din]} C. Ding, \textit{The construction and weight distributions of all projective binary linear codes}, 2020, arXiv: 2010.03184.

\bibitem{[Din2]} C. Ding and H. Niederreiter, \textit{Cyclotomic linear codes of order $3$}, IEEE Trans. Inf. Theory \textbf{56} (2007), 2274-2277.

\bibitem{[DW]} C. Ding and X. Wang, \textit{A coding theory construction of new systematic authentication codes}, Theoretical Computer Science 330 (2005), 81-99.

\bibitem{[FH]} Y. Feng and S. Hong, \textit{Improvements of $T$-adic estimates of exponential sums}, Proc. AMS \textbf{150} (2022), 3687-3698.

\bibitem{[HCP]} Huffman, W. Cary and V. Pless, \textit{Fundamentals of error-correcting codes}, Cambridge university press, 2010.

\bibitem{[Klo]} T. Kl{\o}ve, \textit{Codes for error detection}, Singapore: World Scientific, 2007.


\bibitem{[LN]} R. Lidl and H. Niederreiter, \textit{Finite fields}, second ed., Encyclopedia of Mathematics and its Applications, vol. 20, Cambridge University Press, Cambridge, 1997.

\bibitem{[Moi]} M. Moisio, \textit{A note on evaluations of some exponential sums}, Acta. Arith. \textbf{93} (2000), 117-119.

\bibitem{[Moi2]} M. Moisio, \textit{Explicit evaluation of some exponential sums}, Finite fields Appl. \textbf{15} (2009), 644-651.

\bibitem{[TXF]} C. Tang, C. Xiang and K. Feng, \textit{Linear codes with few weights from inhomogeneous quadratic functions}, Des. Codes Cryptogr. \textbf{83} (2017), 691-714.

\bibitem{[Wan1]} D. Wan, \textit{Variation of $p$-adic Newton polygons for $L$-functions of exponential sums}, Asian J. Math. \textbf{8} (2004), 427-472.

\bibitem{[Wan2]} D. Wan, \textit{Exponential sums over finite fields}, J. Systems Science and Complexity \textbf{34} (2021), 1225-1278.

\bibitem{[WDX]} Q. Wang, K. Ding and R. Xue, \textit{Binary linear codes with two weights}, IEEE Communications Letters \textbf{19} (2015), 1097-1100.
\end{thebibliography}
\end{document}